\documentclass[11pt]{elsarticle}

\usepackage[utf8]{inputenc}

\usepackage{amsmath, amsthm, amssymb}

\usepackage{graphicx}

\usepackage{natbib}

\usepackage{geometry}
\usepackage{setspace}
\usepackage{booktabs}
\usepackage{threeparttable}
\usepackage{longtable}
\usepackage{array}

\usepackage{pifont}
\usepackage{textcomp}
\usepackage{pdflscape}
\usepackage{xcolor}

\usepackage{hyperref}

\usepackage{algorithm}
\usepackage[noend]{algpseudocode}

\newtheorem{prop}{Proposition}

\begin{document}
\begin{frontmatter}

\title{Evaluating Predictive Modeling Strategies for Predicting Individual Treatment Effects in Precision Medicine}

 \author[1]{Pamela Solano}
 \author[2]{M Lee Van Horn}
 \author[1,3]{Thomas Jaki}

 \address[1]{Faculty for Informatics and Data Science, Regensburg University, Germany}
 \address[2]{University of New Mexico, USA}
 \address[3]{MRC Biostatistics Unit, University of Cambridge, UK}

\date{\today}

\begin{abstract}{A structured simulation framework is used to reveal novel differences in model accuracy, directionality, and robustness for predicting individual treatment effects.}

Precision medicine seeks to match patients with treatments that produce the greatest benefit. The Predicted Individual Treatment Effect (PITE)—the difference between predicted outcomes under treatment and control—quantifies this benefit but is difficult to estimate due to unobserved counterfactuals, high dimensionality, and complex interactions. We compared 30+ modeling strategies, including penalized and projection-based methods, flexible learners, and tree-ensembles, using a structured simulation framework varying sample size, dimensionality, multicollinearity, and interaction complexity. Performance was measured using root mean squared error (RMSE) for prediction accuracy and directional accuracy (DIR) for correctly classifying benefit versus harm. Internal validation produced optimistic estimates, whereas external validation with distributional shifts and higher-order interactions more clearly revealed model weaknesses. Penalized and projection-based approaches—ridge, lasso, elastic net, partial least squares (PLS), and principal components regression (PCR)—consistently achieved strong RMSE and DIR performance. Flexible learners excelled only under strong signals and sufficient sample sizes. Results highlight robust linear/projection defaults and the necessity of rigorous external validation.

\end{abstract}


\end{frontmatter}

\section{Introduction} \label{sec:intro}

The estimation of treatment effects has evolved beyond population-level summaries toward individualized approaches, in line with the goals of precision medicine. While randomized controlled trials (RCTs) have traditionally focused on the Average Treatment Effect (ATE), this metric assumes a homogeneous treatment response across individuals. In practice, however, outcomes often vary substantially due to biological, clinical, and sociodemographic factors, resulting in meaningful heterogeneity in treatment response \citep{Kent2018}.

To address this, the concept of Predicted Individual Treatment Effects (PITEs) has received increasing attention as a way to inform patient-specific treatment recommendations \cite{Lamont2016,Chang2021, Jaki2024}. Unlike the ATE, which summarizes average benefit across the population, the PITE estimates the expected difference in predicted outcomes between treatment and control for a given individual, conditional on their covariates. This shift from average to individualized effect estimation represents a fundamental advance in both causal inference and clinical decision-making \citep{Angrist2004, Hoogland2024}. 

Formally, under the PITE framework, let $f_T (x_i)$ denote a predictive function based on the covariate vector $x_i \in \mathbb{R}^p$ for individual $i$, and let $T_i \in \{t, c\}$ be a binary indicator, where $t$ denotes the treatment arm and $c$ control arm. The PITE for subject $i$ is defined as:

\begin{equation}\nonumber
\widehat{\operatorname{PITE}_i} = \hat{f}_t(x_i) - \hat{f}_c(x_i), \quad 1 \leq i \leq n,
\end{equation}
where $\hat{f}_t(\cdot)$ and $\hat{f}_c(\cdot)$ are predicted models fitted separately to the treatment and control groups, respectively. 

Under this construction, assuming that the first two moments of these predictor functions are defined, the variance of $\widehat{\mathrm{PITE}}_i$ is

$$\mathrm{Var(PITE_i)} = \operatorname{Var}{(\hat f_t(x_i))}+ \operatorname{Var}{(\hat f_c(x_i))} 
$$ 
where the covariance term vanishes, i.e.,  $\operatorname{Cov}(\hat{f}_t, \hat{f}_c) = 0$, since the models for the two arms are estimated independently.

A wide range of methods can be used to estimate PITE, motivated only by rendering uncorrelated $\hat{f}_t$ and $\hat{f}_c$. However, external challenges appear, such as dimensionality, heterogeneity, nonlinear outcomes. Data-adaptive techniques such as penalized regression \citep{Hastie-Zou2005}, tree-based ensembles \citep{Breiman2001, Wager2018}, and Bayesian approaches \citep{Chipman2010, Hill2011} have become increasingly popular, particularly in early-phase trials and real-world studies where sample sizes are limited and covariate structures are complex. In the PITE setting, prior work \cite{Ballarini2018} has highlighted several challenges. Ultimately, defining what constitutes good performance for PITE estimation is highly context-dependent, and no single approach can be considered universally optimal. 

Despite substantial methodological progress, accurately estimating PITEs remains dificult, with little consensus on which methods or assumptions are most reliable. A key difficulty is that models performing well in outcome prediction do not necessarily provide valid or stable estimates of treatment effect. Minimizing predictive loss (e.g., mean squared error for observed outcomes) does not directly minimize the error in treatment effects estimation \citep{Dorie2016}. Analogous reasoning applies to PITE estimation, as demonstrated in Proposition \ref{prop:mse}.

A common but misleading assumption is that if the potential outcomes under treatment and control are accurately estimated, the resulting PITE estimates
must also be accurate. While intuitive, this assumption fails both algebraically and empirically. Because PITE is defined as a contrast between two predictors, the errors from each model can interact in unpredictable ways—magnifying noise, introducing bias, or distorting calibration. Several studies \citep{Rolling2014, Zhao2017, Efthimiou2023} have illustrated that even when outcome models perform well individually, the derived treatment effect can be unstable, biased, or poorly calibrated. Importantly, none of the standard performance metrics for PITE (e.g., RMSE, MAE, R², calibration slope) can be directly inferred from outcome-level performance measures. 

We formalize this distinction within a general MSE-based framework and show that models achieving high predictive accuracy for the outcome $Y$ may nonetheless yield unstable or misleading PITEs (Section~\ref{sec:diag}). 

In light of these concerns, this paper does not aim to promote a single best method. Instead, we offer a principled and practical guide to the reliable estimation of PITEs, with a focus on identifying the statistical properties, limitations, and trade-offs that arise when modeling treatment effects in clinical trial settings. We emphasize understanding how estimation choices influence variability, bias, calibration, and decision relevance in clinical trial settings.

To support this aim, we conduct a comprehensive and structured evaluation of over 30 representative methods across five model families: (i) regularization-based models (e.g., Lasso, Elastic Net, Bayesian shrinkage), (ii) tree-based ensemble learners (e.g., Random Forest, Gradient Boosting, BART \citep{Hill2011}), (iii) Bayesian learning approaches (e.g., spike-and-slab priors, Bayesian Neural Networks), (iv) traditional parametric models (e.g., GLMs with interaction terms), and (v) nonlinear and kernel-based methods (e.g., neural networks, KRLS, PPR).

Despite increasing enthusiasm for individualized treatment effects inference, fundamental scientific questions remain unresolved regarding the reliable estimation and validation of PITEs. How do modeling choices—such as regularization, projection, or nonparametric flexibility—interact with sample size, dimensionality, and treatment effect heterogeneity to influence bias, variance, and calibration of the resulting estimates? Under what data conditions do flexible learners, including tree ensembles and Bayesian additive models, genuinely outperform traditional penalized or projection-based approaches, and when do they instead risk overfitting or instability? Equally important, how should predictive performance be quantified when the estimand itself (the counterfactual contrast) is unobservable, and to what extent do existing metrics, such as RMSE and directional accuracy, capture clinically meaningful reliability?

This paper addresses these questions through a systematic, simulation-based evaluation of over thirty contemporary modeling strategies within the PITE framework. Our objective is not merely to rank methods by performance but to uncover the principles that govern when and why certain models succeed or fail in representing individualized treatment effects. By explicitly linking model adaptability to patient-level complexity, we aim to clarify the conditions under which variation in predicted treatment effects across patients translates into actionable clinical insight. By explicitly linking model adaptability to patient-level complexity, we aim to clarify the conditions under which variation in predicted treatment effects across patients translates into actionable clinical insight. These questions form the foundation of a broader agenda: developing robust, interpretable\footnote{Interpretability refers to transparency in the model’s statistical structure—its parameterization, shrinkage behavior, and variance–bias tradeoff characteristics—rather than interpretability of the predicted individual treatment effects themselves.} and generalizable tools for predicting individual treatment effects in precision medicine.

To address these questions, we designed a comprehensive and rigorous simulation framework that mirrors the complexities of real-world clinical data. The evaluation spans scenarios with varying degrees of treatment effect heterogeneity, high-dimensional and correlated covariates, limited sample sizes, and deliberate model misspecification—conditions under which traditional estimators often falter. Each modeling strategy is examined not only for its predictive accuracy but also for its inferential stability and clinical interpretability. Performance is quantified using complementary metrics that capture distinct dimensions of reliability: the expected mean squared error (MSE) of the PITE for overall estimation quality, calibration and adjusted $R^2$ for model fit, mean absolute error (MAE) for robustness to outliers, and a directional accuracy criterion assessing whether models correctly classify treatment benefit versus harm, independent of effect magnitude. 

We recognize that direct validation of counterfactuals is impossible, thus we also emphasize the importance of robust performance evaluation techniques. These include matching-based validation frameworks \citep{Gao2021, Kuhlemeier_2024} and metrics assessing calibration and discrimination \citep{VanCalster2018, Hoogland2024}, which are critical for regulatory, clinical, and external transportability assessments \citep{Steyerberg2019}. Furthermore, we address common real-world complexities—such as missing data, outliers, and complex biomarker interactions—using tools like multiple imputation \citep{Lamont2016}, Bayesian data augmentation, and regularization techniques that help stabilize PITE estimates in high-dimensional settings \citep{Ballarini2018}.

This work contributes in three primary ways. First, it classifies and contextualizes the diverse set of estimation strategies available for PITEs, linking them to their underlying assumptions and use cases. Second, it provides empirical insights into how estimation strategies perform under different data-generating conditions. Third, it delivers practical recommendations and cautionary guidance, including ``do'' and ``avoid'' principles that emphasize when outcome-focused loss functions (e.g., MSE) may lead to invalid or unstable treatment effect estimates.

\section{Methodology}

\subsection{Evaluated Approaches}
We evaluated a broad and representative family of models for estimating predicted individual treatment effects (PITEs), spanning classical statistical approaches, modern machine learning, and Bayesian methods. Linear and regularized linear models—including Ordinary Least Squares \citep{seber2004linear}, Generalized Linear Models (GLMs) \citep{mccullagh1989generalized}, Lasso \citep{Tibshirani1996}, Ridge \citep{Hoerl2000}, Elastic Net \citep{Hastie-Zou2005}, and Principal Component Regression \citep{Jolliffe_PCA2002}—offer a well-established foundation, particularly valued for interpretability and computational efficiency. These models are effective in low to moderate dimensions but may struggle with nonlinearities and complex feature interactions. To address multicollinearity and high-dimensionality, we incorporated penalized and projection-based approaches such as Partial Least Squares \citep{Wold1987}, stepwise selection \citep{Efron2004}, and non-negative constraints via Non-Negative Least Squares \citep{Lawson1995}. These techniques allow for parsimonious modeling and feature selection while maintaining interpretability.

To model nonlinear and interaction-rich relationships, we included tree-based and ensemble learners (e.g., Decision Trees, Random Forests \citep{Breiman2001}, Boosted Trees \citep{Friedman2001}, Quantile Forests \citep{meinshausen2006quantile}), which are particularly effective for high-dimensional data without requiring strict parametric assumptions. Bayesian variants such as Bayesian Additive Regression Trees (BART) \citep{Chipman2010} and Spike-and-Slab regression \citep{Brown1998} enable uncertainty quantification and adaptive regularization, while neural networks—including Bayesian Neural Networks \citep{MacKay1992}—capture complex patterns and offer robust predictive performance, albeit at higher computational cost. We also considered flexible non-parametric and kernel-based methods, such as Kernel Ridge Regression \citep{tipping2001} and Projection Pursuit Regression \citep{Friedman1981}, as well as constrained and latent variable models like Quantile Regression \citep{Koenker2005} and Supervised PCA \citep{Bair2006}. These models bring strengths in capturing distributional nuances and handling challenging data structures. Altogether, this modeling spectrum reflects a balance between flexibility, interpretability, and computational tractability, suitable for diverse clinical data scenarios encountered in PITE estimation. Further detail and abbreviations for names of this methods are is provided in Supplemental Material, Section~\ref{app: xapproach}. 

Before defining our PITE evaluation metrics, we first examined them from the perspective of \textit{explained uncertainty} and \textit{interpretability}. This analysis highlights the complementary roles of \textbf{RMSE} and \textbf{Direction}, which are the diagnostic metrics adopted in our main analysis. 

\subsection{Diagnostic Predictive Performance}\label{sec:diag}

We assess the predictive performance of \emph{PITE estimates} using two complementary metrics that capture different aspects of accuracy and clinical utility.

The first is the \textbf{Root Mean Squared Error (RMSE)}, which summarizes the overall discrepancy between predicted and true (unobserved) predicted individual treatment effects:
\begin{equation}
\mathrm{RMSE} = \sqrt{\mathbb{E}\left[(\widehat{PITE}(x_i) - PITE(x_i))^2\right]}.
\end{equation}

RMSE reflects both variance and bias in the PITE estimates and serves as a global measure of estimation accuracy.

The second metric, referred to as \textbf{Direction}, quantifies the ability to correctly identify the sign of the treatment effects for each individual:
\begin{equation}\label{eq:direction}
\mathrm{Direction} = \frac{1}{n} \sum_{i=1}^n \mathbb{I}\!\left(\mathrm{sign}(\widehat{PITE}(x_i)) = \mathrm{sign}(PITE(x_i))\right).
\end{equation}

where $\mathbb{I}(\cdot)$ denotes the indicator function. This metric measures agreement in effect direction rather than magnitude, which is particularly important in decision-making contexts such as precision medicine, where identifying whether a treatment is beneficial or harmful may outweigh precise effect estimation.

While both metrics are intuitive, their behavior is best understood through a formal decomposition of the PITE error. Specifically, the Mean Squared Error of PITE can be expressed in terms of the MSE associated with the treatment and control outcome models:

\begin{prop}\label{prop:mse} Assume that $\hat{f}_t(\cdot)$ and $\hat{f}_c(\cdot)$ are predictors for which an expected value and variance can be well defined. Then the PITE Mean Squared Error ($\operatorname{MSE}_{\mathrm{PITE}}$) can be decomposed as 

\begin{equation}
\operatorname{MSE}_{\operatorname{PITE}} = \operatorname{MSE}_{\operatorname{t}} + \operatorname{MSE}_{\operatorname{c}} - 2\,\mathrm{bias}_t \cdot \mathrm{bias}_c,
\end{equation}
 % - 2\, \operatorname{Cov}[\hat{f}_1(x), \hat{f}_0(x)]
where $\operatorname{MSE}_{\operatorname{t}} = \mathrm{Var}[\hat{f}_t(X)] + \mathrm{bias}_t^2$ and $\operatorname{MSE}_{\operatorname{c}} = \mathrm{Var}[\hat{f}_c(X)] + \mathrm{bias}_c^2$ denote the mean squared errors for the treatment and control models, respectively, with $\mathrm{bias}_t = \mathbb{E}[\hat{f}_t(X)] - \mathbb{E}[Y(t)]$ for treatment and $\mathrm{bias}_c = \mathbb{E}[\hat{f}_c(X)] - \mathbb{E}[Y(c)]$ for control.
\end{prop}

This proposition formalizes the intuition that minimizing predictive loss for the outcome models $\hat{f}_t$ and $\hat{f}_c$ does not necessarily minimize the error in PITE estimation. This decomposition highlights that PITE accuracy depends not only on the individual predictive performance of $\hat{f}_t$ and $\hat{f}_c$, but also on the interaction of their biases. In particular, biases in opposite directions can compound and amplify errors in PITE estimates, underscoring the need for diagnostics that account for the joint behavior of the two outcome models.

For completeness, we also evaluated additional performance metrics to complement the primary error measures. Let $\text{PITE}_i$ denote the true predicted individual treatment effect for subject $i$, and $\hat{\text{PITE}}_i$ its corresponding estimate obtained from a given model. 

The \textbf{Mean Absolute Error (MAE)} summarizes the average absolute deviation between the estimated and true PITEs:
\[
\mathrm{MAE}_{\text{PITE}} = \frac{1}{n} \sum_{i=1}^n \big|\hat{\text{PITE}}_i - \text{PITE}_i\big|.
\]

The \textbf{Coefficient of Determination} ($R^2$) quantifies the proportion of variation in the true PITE values explained by the model estimates:
\[
R^2_{\text{PITE}} = 1 - 
\frac{\sum_{i=1}^n (\text{PITE}_i - \hat{\text{PITE}}_i)^2}
{\sum_{i=1}^n (\text{PITE}_i - \bar{\text{PITE}})^2},
\]
where $\bar{\text{PITE}}$ is the sample mean of the true PITE values.

Finally, we assessed \textbf{calibration} using a simple linear regression of the true effects on the estimated ones:
\[
\text{PITE}_i = \alpha + \beta\,\hat{\text{PITE}}_i + \eta_i,
\]
where $\alpha$ and $\beta$ denote the intercept and slope, respectively. Perfect calibration corresponds to $\alpha = 0$ and $\beta = 1$, indicating that estimated effects are unbiased in both level and scale. Deviations from these values reflect systematic bias (nonzero $\alpha$) or miscalibration in effect magnitude (slope $\beta \neq 1$).

Detailed definitions and decompositions of these metrics are provided in Supplemental Material Section~\ref{app: metrics}, along with an empirical justification for our focus on RMSE and Direction as primary evaluation criteria rather than the other metrics.

\section{Simulation Study}\label{sec:simulation}

\subsection{Simulation Design}

We generate synthetic data to evaluate the performance of the proposed PITE framework under controlled conditions. For each replication, we simulate $N$ independent subjects with $p$-dimensional covariate vectors $X_i = (X_{i1}, \dots, X_{ip})^\top$ drawn from a multivariate normal distribution:
\[
X_i \sim \mathcal{N}_p(0, \Sigma), \qquad i = 1,\dots,N,
\]
where $\Sigma$ is a $p \times p$ covariance matrix specifying the correlation structure among covariates. Unless otherwise stated, $\Sigma$ is set to $\rho I_p$ for simplicity, where $\rho$ controls the correlation magnitude.

Treatment assignment is independent of covariates and is allocated to half of the subjects, where $T_i = 1$ indicates assignment to the treatment arm and $T_i = 0$ to the control arm.

We specify the potential outcomes under control and treatment as follows. For each subject, the control mean response is given by
\[
f_0(X_i) = X_i^\top \beta_0,
\]
where $\beta_0 \in \mathbb{R}^p$ is a vector of baseline coefficients drawn from $\mathcal{N}(0, 0.1^2)$.

The individual treatment effect is modeled as
\[
\Delta(X_i) = X_i^\top \beta_\Delta,
\]
where $\beta_\Delta \in \mathbb{R}^p$ is a vector of treatment effect coefficients that determine how covariates modulate the individual treatment benefit. In our simulations, $\beta_\Delta$ is drawn from
\[
\beta_\Delta \sim \mathcal{N}(\mu_{\beta_\Delta}, 0.01^2 I_p),
\]
where $\mu_{\beta_\Delta}$ specifies the mean of the treatment effect coefficients. The small variance $0.01^2$ induces only mild heterogeneity, ensuring that differences in PITE estimates are primarily driven by changes in the mean treatment effect rather than high variability across individuals. This construction allows the expected individual benefit to mantain modest, symmetric perturbations around the target effect.

The observed outcome is generated from the potential outcomes framework with additive Gaussian noise:
\[
Y_i(c) = f_c(X_i) + \varepsilon_{ci}, \qquad Y_i(t) = Y_i(c) + \Delta(X_i),
\]
where $\varepsilon_{ci} \sim \mathcal{N}(0,1)$. Finally, the observed response is
\[
Y_i^{\mathrm{obs}} = Y_i(c) + T_i \cdot \Delta(X_i).
\]

The simulated dataset therefore includes covariates $(X_i)$, treatment indicator $(T_i)$, observed outcome $(Y_i^{\mathrm{obs}})$, and the true individual expected treatment effect $\Delta(X_i)$, which serves as the ground truth for evaluation. Simulations were executed in parallel across 27 computational threads replicated 135 times. Details regarding software and model implementation are provided in Supplemental Material Section~\ref{app:sim_details}. The exact algorithm is detailed in Algorithm \ref{alg:sim_external_validation}.

%Specifically, we vary sample sizes, numbers of features, interactions and correlations among covariates affecting the benefit component,  across tree levels of expected treatment effect (ete): 0 (low), 0.25 (moderate), and 0.5 (high) and testing under internal and external validation strategy were evaluated. 

%Models were considered successful when they exhibited both high directional accuracy (DIR) and low {Root Mean Squared Error (RMSE)}.

\begin{algorithm}[!htbp]
\caption{Simulation–Validation–Matching Pipeline for PITE Evaluation}
\label{alg:sim_external_validation}
\begin{algorithmic}[1]
\Statex \textbf{Input:} Simulation parameters $\{n, p, \text{ete}\}$; list of predictive models $\mathcal{M} = \{m_1, \dots, m_K\}$; data generator $\texttt{sim\_data()}$; number of CV folds $=10$.
\Statex \textbf{Output:} Estimated PITEs $\widehat{\Delta}_{ij}$, observed PITEs $O_{ij}$, and true benefits for training and validation populations.

\vspace{2mm}
\Statex \textbf{Step 1: Data Generation (Training Dataset)}
\State Simulate the training dataset $\mathcal{D}_{\text{train}} = \{(\boldsymbol{X}_i, T_i, Y_i)\}_{i=1}^{n}$ using $\texttt{sim\_data}(p, n, \text{ete})$.
\State Partition $\mathcal{D}_{\text{train}}$ into two treatment subsets: 
$$\mathcal{D}_{0} = \{(\boldsymbol{X}_i, Y_i): T_i = 0\}  \mbox{ and }  \mathcal{D}_{1} = \{(\boldsymbol{X}_i, Y_i): T_i = 1\}$$

\Statex \textbf{Step 2: External Dataset (Generation)}  
\State Independently generate a new dataset $\mathcal{D}_{\text{val}} = \{(\boldsymbol{X}^*_j, T^*_j, Y^*_j)\}_{j=1}^{n}$ using the same data-generating mechanism but with a different random seed, thereby producing an \emph{independent external population} that shares the same structural parameters but distinct realizations of covariates, $\mu_{\beta_{\Delta}}$, treatment assignments, and outcomes.  
\State Partition $\mathcal{D}_{\text{val}}$ into treatment ($T^* = 1$) and control ($T^* = 0$) subgroups.

\vspace{2mm}
\Statex \textbf{Step 3: Matching Procedure on External Validation Dataset for Prediction}
\State Apply nearest-neighbor matching using Mahalanobis distance between treated and control subjects in $\mathcal{D}_{\text{val}}$ via \texttt{MatchIt}.
\State Obtain index sets $(\mathcal{I}_A, \mathcal{I}_B)$ corresponding to matched pairs: $$\mathcal{I}_A = \{i: T^*_i = 0\} \mbox{ and } \mathcal{I}_B = \{j: T^*_j = 1\}$$
\State These indices are used to compute pairwise observed and predicted treatment-effect differences between matched subjects.

\vspace{2mm}
\Statex \textbf{Step 4: Model Training and Prediction}
\State Each model $m_k \in \mathcal{M}$ is trained separately on $\mathcal{D}_{0}$ and $\mathcal{D}_{1}$ to predict potential outcomes under control and treatment, respectively.
    \[
    \widehat{f}_{0}^{(k)} = m_k(Y \sim \boldsymbol{X}, \mathcal{D}_0), \quad 
    \widehat{f}_{1}^{(k)} = m_k(Y \sim \boldsymbol{X}, \mathcal{D}_1).
    \]
\State Use 10-fold cross-validation for models.%; method-specific hyperparameters are used for others (e.g., BART: \texttt{nskip=50}, \texttt{ndpost=200}; Bayenet: \texttt{max.steps=500}, \texttt{penalty="lasso"}).
 \State Predict potential outcomes for each validation subject following \textbf{Step 3}:
    \[
    \widehat{Y}_{0j}^{(k)} = \widehat{f}_{0}^{(k)}(\boldsymbol{X}^*_j), \quad 
    \widehat{Y}_{1j}^{(k)} = \widehat{f}_{1}^{(k)}(\boldsymbol{X}^*_j).
    \]
\State Compute estimated PITEs for matched pairs:
    \[
    \widehat{\Delta}_{ij}^{(k)} = \widehat{Y}_{1,\mathcal{I}_B}^{(k)} - \widehat{Y}_{0,\mathcal{I}_A}^{(k)}.
    \]
\State Compute observed PITEs:
    \[
    O_{ij} = Y^*_{\mathcal{I}_B} - Y^*_{\mathcal{I}_A}.
    \]
\State Store $\widehat{\Delta}_{ij}^{(k)}$ and $O_{ij}$ for later evaluation.

\vspace{2mm}
\Statex \textbf{Remarks:}
\begin{itemize}
    \item Training and validation are performed on independent simulated populations (step applied only for external validation).
    \item Matching is applied only to the validation dataset to form comparable treated–control pairs for evaluation.
\end{itemize}

\end{algorithmic}
\end{algorithm}

\subsection{Internal Validation}\label{sec:successfull_int}

Covariates were sampled from a multivariate normal distribution with a specified correlation structure to represent different levels of multicollinearity: low ($\rho = 0$), moderate ($\rho = 0.5$), and high ($\rho = 0.95$). Treatment assignment was independent of covariates and half of the subjects were assigned to the treatment group and half to the control group. Outcomes were generated as linear functions of the covariates with additive Gaussian noise. 

Across simulation scenarios, we varied the mean of the treatment effect coefficients ($\mu_{\beta_\Delta} \in \{0, 0.25, 0.5\}$), the number of covariates ($p \in \{5, 15, 45\}$), and the sample size ($n \in \{250, 500, 750\}$).

Within a single population, we generated both training and test datasets by splitting the sample in a 50:50 ratio, following the approach of \citet{Rolling2014}. Predictive models were then trained on the training set and evaluated on the test set.

%%%%% general overview %%%%%%%%%
Under internal validation with correlated covariates, varying sample sizes, and a range of treatment effects, model performance was generally strong and more stable compared to external validation with interactions. The average RMSE was around (mean = 0.88, median = 0.51) in the range of (0.03, 22), and directional accuracy (DIR) was high overall (mean = 0.76, median = 0.84), with several models achieving near-perfect directionality (DIR $> 0.95$). 

\begin{figure}[h!]
    \centering
    \includegraphics[width=1\linewidth]{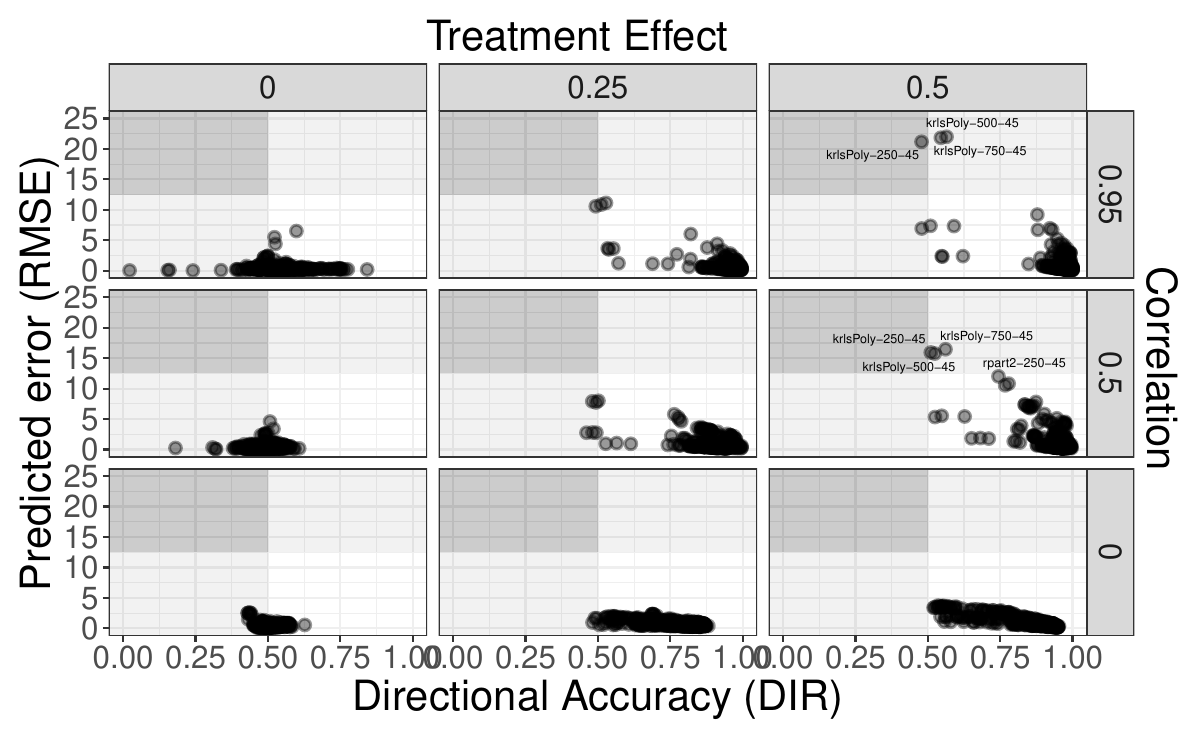}
    \caption{Internal validation: Relationship between predictive performance Root Mean Squared Error (RMSE) and Directional Accuracy (DIR) across models and Treatment Effect impact and correlation value between covariates. Each point represents a specific model-parameter configuration, labeled by model name, sample size (n), and number of predictors (p). Darkness gray areas reflect undesirable regions}
    \label{fig:internal_risk_dir_cor}
\end{figure}

%In the internal validation setting with correlated structures among covariates, 
Figure~\ref{fig:internal_risk_dir_cor} reveals clear performance variability across models, measured by RMSE and DIR. Several models, notably krlsPoly and rpart2, were repeatedly located within the grey danger zones, indicating high RMSE and low DIR. This pattern was especially pronounced in scenarios with high-dimensional covariates ($p = 45$) and low or null treatment effects, where predictive accuracy deteriorated and directional agreement failed to exceed acceptable thresholds ($\operatorname{DIR} < 0.6$). Conversely, well-performing models remained outside these danger areas, achieving better alignment between estimated and true individual treatment effects.

% 
%%%%% worst cases %%%%%%%%%

%   dplyr::filter(RMSE>=5 & DIR <= 0.8)%>% 

Even under internal validation with correlated covariates---methods demonstrated substantial predictive error and poor directional reliability in estimating PITE. Notably, krlsPoly and rpart2 frequently appeared among the worst performers, with RMSE values reaching up to 22.0, MAE exceeding 17.5, the mayimum R² achieves 79.7\%, poor $\alpha$ \& $\beta$ coverage and DIR consistently near or below 0.55. These failures were most prominent in high-dimensional settings ($p = 45$) and especially under moderate to high correlation ($\rho = 0.5–0.95$), across all treatment effect levels, Table~\ref{tab:internal_ext_worst_corr}. 

\newpage
% latex table generated in R 4.3.1 by xtable 1.8-4 package
% Thu Aug  7 18:27:16 2025
%\begin{table}[ht]
\begin{table}[H]
\centering
\begin{tabular}{rrlrrrrrrrr}
  \hline
$\mu_{\beta_\Delta}$ & $\rho$ & model & p & n & RMSE & R² & MAE & DIR & $\alpha_0$ & $\alpha_1$ \\ 
  \hline
0.00 & 0.95 & Bayenet & 45 & 250 & 6.507 & 0.148 & 6.233 & 0.598 & 0.400 & 0.200 \\ 
  0.00 & 0.95 & Bayenet & 45 & 500 & 5.492 & 0.384 & 5.314 & 0.523 & 0.000 & 0.000 \\ 
  0.25 & 0.50 & krlsPoly & 45 & 250 & 8.002 & 0.706 & 6.367 & 0.502 & 0.000 & 0.000 \\ 
  0.25 & 0.50 & rpart2 & 45 & 250 & 5.784 & 0.494 & 4.571 & 0.764 & 1.000 & 0.000 \\ 
  0.25 & 0.50 & krlsPoly & 45 & 500 & 7.904 & 0.681 & 6.286 & 0.481 & 0.000 & 0.000 \\ 
  0.25 & 0.50 & rpart2 & 45 & 500 & 5.304 & 0.565 & 4.147 & 0.778 & 0.400 & 0.000 \\ 
  0.25 & 0.50 & krlsPoly & 45 & 750 & 7.729 & 0.735 & 6.227 & 0.494 & 0.000 & 0.000 \\ 
  0.25 & 0.95 & krlsPoly & 45 & 250 & 11.134 & 0.341 & 8.945 & 0.528 & 0.000 & 0.000 \\ 
  0.25 & 0.95 & krlsPoly & 45 & 500 & 10.574 & 0.638 & 8.523 & 0.493 & 0.000 & 0.000 \\ 
  0.25 & 0.95 & krlsPoly & 45 & 750 & 10.861 & 0.611 & 8.545 & 0.511 & 0.000 & 0.000 \\ 
  0.50 & 0.50 & krlsPoly & 15 & 250 & 5.327 & 0.616 & 4.284 & 0.525 & 0.000 & 0.000 \\ 
  0.50 & 0.50 & krlsPoly & 15 & 500 & 5.537 & 0.653 & 4.422 & 0.548 & 0.000 & 0.000 \\ 
  0.50 & 0.50 & krlsPoly & 15 & 750 & 5.423 & 0.674 & 4.328 & 0.627 & 0.000 & 0.000 \\ 
  0.50 & 0.50 & krlsPoly & 45 & 250 & 15.935 & 0.759 & 12.688 & 0.510 & 0.000 & 0.000 \\ 
  0.50 & 0.50 & rpart2 & 45 & 250 & 12.009 & 0.498 & 9.797 & 0.744 & 0.800 & 0.000 \\ 
  0.50 & 0.50 & krlsPoly & 45 & 500 & 15.762 & 0.801 & 12.466 & 0.524 & 0.000 & 0.000 \\ 
  0.50 & 0.50 & rpart2 & 45 & 500 & 10.535 & 0.568 & 8.431 & 0.767 & 1.000 & 0.000 \\ 
  0.50 & 0.50 & krlsPoly & 45 & 750 & 16.459 & 0.797 & 13.197 & 0.561 & 0.000 & 0.000 \\ 
  0.50 & 0.50 & rpart2 & 45 & 750 & 10.822 & 0.577 & 8.613 & 0.779 & 0.600 & 0.000 \\ 
  0.50 & 0.95 & krlsPoly & 15 & 250 & 7.327 & 0.541 & 5.832 & 0.590 & 0.000 & 0.000 \\ 
  0.50 & 0.95 & krlsPoly & 15 & 500 & 6.925 & 0.703 & 5.517 & 0.479 & 0.000 & 0.000 \\ 
  0.50 & 0.95 & krlsPoly & 15 & 750 & 7.362 & 0.698 & 5.887 & 0.509 & 0.000 & 0.000 \\ 
  0.50 & 0.95 & krlsPoly & 45 & 250 & 21.173 & 0.519 & 16.963 & 0.478 & 0.000 & 0.000 \\ 
  0.50 & 0.95 & krlsPoly & 45 & 500 & 21.999 & 0.577 & 17.548 & 0.565 & 0.000 & 0.000 \\ 
  0.50 & 0.95 & krlsPoly & 45 & 750 & 21.833 & 0.757 & 17.453 & 0.545 & 0.000 & 0.000 \\ 
   \hline
\end{tabular}
\caption{Internal Validation: Scenarios where the Root Mean Squared Error ($\operatorname{RMSE} \geq 5$) and Directional Accuracy ($\operatorname{DIR} \leq 0.8$), reflecting poor performance in estimation of personalized treatment effects (PITE).}\label{tab:internal_ext_worst_corr}
\end{table}

%%%%% best cases comparison %%%%%%%%%
\paragraph{Successful Model Performance Across Simulated Scenarios with Internal Validation}

\begin{figure}
    \centering
    \includegraphics[width=1\linewidth]{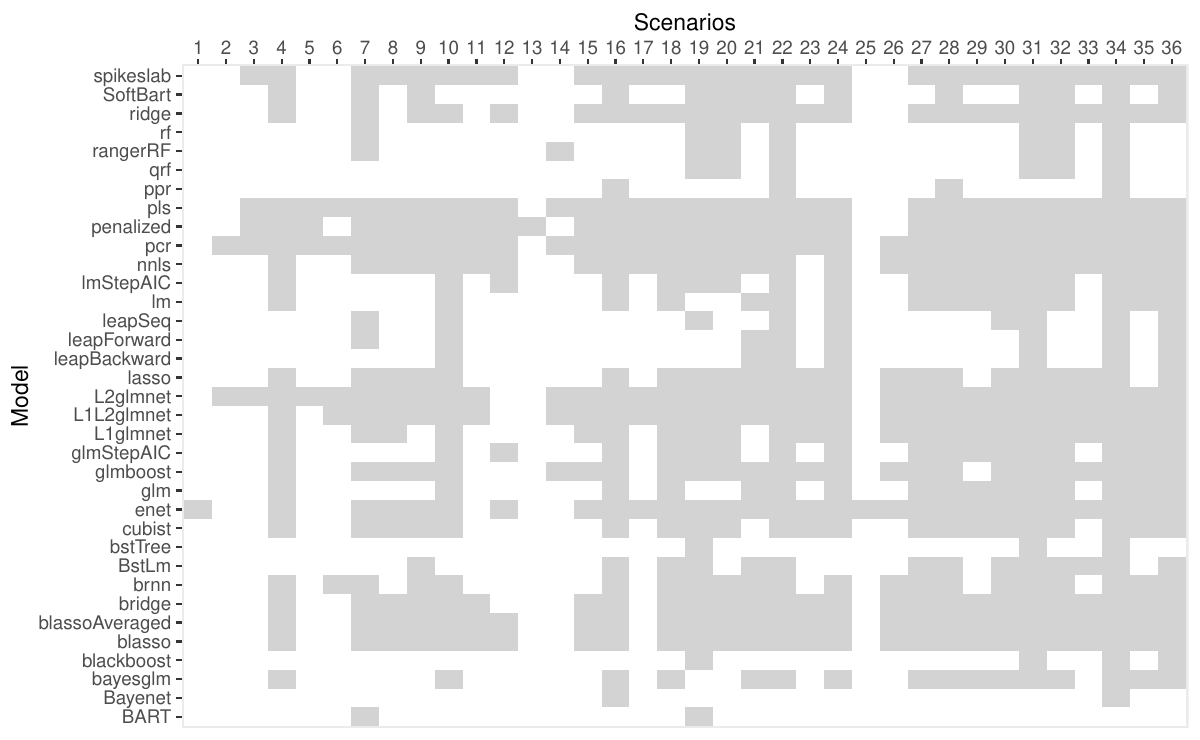}
    \caption{Internal Validation: Models Achieving High PITE Accuracy ($\operatorname{RMSE} < 1$ and $\operatorname{DIR} > 0.95$) Across External Validation Scenarios with Covariate Interactions}\label{fig:internal_best_risk_dir_corr}
        \vspace{0.5em}
\begin{center}        
\parbox{\textwidth}
{\small
\begin{minipage}{\linewidth}
\setlength{\tabcolsep}{3.5pt} % tighter columns
\textit{setting}$n$ = sample size; $\rho$ = covariate correlation; $\mu_{\beta_\Delta}$ = mean of the treatment effect coefficients; $p$ = number of covariates.\\[2.5pt]
\begin{tabular}{rcccc|rcccc|rcccc|rcccc}
\hline
ID & $n$ & $\rho$ & $\mu_{\beta_\Delta}$ & $p$ & 
ID & $n$ & $\rho$ & $\mu_{\beta_\Delta}$ & $p$ & 
ID & $n$ & $\rho$ & $\mu_{\beta_\Delta}$ & $p$ & 
ID & $n$ & $\rho$ & $\mu_{\beta_\Delta}$ & $p$ \\
\hline
1  & 250 & 0    & 0.50 & 5  & 10 & 250 & 0.95 & 0.50 & 15 & 19 & 500 & 0.95 & 0.25 & 15 & 28 & 750 & 0.5  & 0.50 & 15 \\
2  & 250 & 0.5  & 0.25 & 15 & 11 & 250 & 0.95 & 0.50 & 45 & 20 & 500 & 0.95 & 0.25 & 45 & 29 & 750 & 0.5  & 0.50 & 45 \\
3  & 250 & 0.5  & 0.25 & 45 & 12 & 250 & 0.95 & 0.50 & 5  & 21 & 500 & 0.95 & 0.25 & 5  & 30 & 750 & 0.5  & 0.50 & 5  \\
4  & 250 & 0.5  & 0.50 & 15 & 13 & 500 & 0    & 0.50 & 5  & 22 & 500 & 0.95 & 0.50 & 15 & 31 & 750 & 0.95 & 0.25 & 15 \\
5  & 250 & 0.5  & 0.50 & 45 & 14 & 500 & 0.5  & 0.25 & 15 & 23 & 500 & 0.95 & 0.50 & 45 & 32 & 750 & 0.95 & 0.25 & 45 \\
6  & 250 & 0.5  & 0.50 & 5  & 15 & 500 & 0.5  & 0.25 & 45 & 24 & 500 & 0.95 & 0.50 & 5  & 33 & 750 & 0.95 & 0.25 & 5  \\
7  & 250 & 0.95 & 0.25 & 15 & 16 & 500 & 0.5  & 0.50 & 15 & 25 & 750 & 0    & 0.50 & 15 & 34 & 750 & 0.95 & 0.50 & 15 \\
8  & 250 & 0.95 & 0.25 & 45 & 17 & 500 & 0.5  & 0.50 & 45 & 26 & 750 & 0.5  & 0.25 & 15 & 35 & 750 & 0.95 & 0.50 & 45 \\
9  & 250 & 0.95 & 0.25 & 5  & 18 & 500 & 0.5  & 0.50 & 5  & 27 & 750 & 0.5  & 0.25 & 45 & 36 & 750 & 0.95 & 0.50 & 5  \\
\hline
\end{tabular}
\end{minipage}
}
\end{center}
\end{figure}

Figure~\ref{fig:internal_best_risk_dir_corr} evidences that under internal validation with correlated covariates the landscape of success is reassuring but structured. In scenarios with strong performance ($\operatorname{RMSE} < 1$ and $\operatorname{DIR} > 0.95$): simple, penalized and projection-based linear methods are the most consistently successful approaches based on lower RMSE and higher DIR. Methods such as the glm-family (glm, bayesglm), penalized regressions (lasso, ridge, enet, L1/L2glmnet/L2glmnet), stepwise/selection variants (glmStepAIC, lmStepAIC, leapForward/Backward/Seq), and low-dimensional projection methods (pls, pcr, nnls, penalized) appear repeatedly across a very wide array of conditions---from small samples and weak covariate correlation to large samples and strong correlation (representative condition indices: [1, 6, 12, 18, 24, 30, 36]). 

By contrast, more flexible, nonparametric and tree/ensemble methods (BART/SoftBart, RF family, QRF, Cubist, boosted trees like blackboost/bstTree) succeed only under narrower conditions---typically when sample size, signal expected treatment effect and/or covariate correlation are favorable. For example, BART seems to work well under high-correlation moderate-signal settings ([7, 19] conditions — i.e. $\rho = 0.95$ with moderate expected treatment effect), while the random-forest family (rf, rangerRF, qrf) tends to appear in the successful set mainly in the high-correlation strata (e.g. [7, 19, 31]). 

\subsection{External Validation}

This second simulation scenario assesses the extent to which models trained in one population generalize when applied to an external population—a critical yet often underexamined challenge in practical applications. The simulation design follows the procedure described in \citet{Kuhlemeier_2024} and is detailed in Algorithm~\ref{alg:sim_external_validation}. Simulation parameters governing the covariate structure, including $p$, $n$, $\rho$, and $\mu_{\beta_{\Delta}}$, were varied as in the previous simulation scenario.

%\section{Results}
RMSE values ranged from a minimum of 1.34 to a maximum of 22.23, with a median of 1.56. The Direction (DIR) metrics ranged from 0.16 to 0.98 (mean = 0.69, median = 0.70), indicating good directional performance and clinical interpretability in most cases.

\subsubsection{Correlated Covariates}

Figure \ref{fig:Ext_risk_sensit_cor_te}  illustrates the trade-off between prediction error (RMSE) and directional accuracy (DIR) across a range of model configurations, treatment effect sizes, and covariate correlation structures. Predictive performance generally deteriorated with increasing correlation among covariates ($\rho$), with models under $\rho = 0.95$ exhibiting higher RMSE and lower DIR. As the mean of the treatment effect coefficients increased from 0 to 0.5, DIR improved due to stronger signal, but RMSE also rose in some cases, reflecting the increased difficulty of precise estimation. Models with larger sample sizes and moderate dimensionality (e.g., n = 750, p = 45) consistently achieved favorable performance (low RMSE, high DIR), particularly when covariates were weakly correlated. In contrast, models with high dimensionality and small sample sizes frequently failed to identify individual-level treatment effects accurately. In particular, the KRLS with polynomial basis expansions (krlsPoly) had its accuracy decreased with stronger correlation among covariates, with $\rho$ = 0.95 leading to elevated RMSE and reduced DIR. Even with sample $n=750$, high-dimensional configurations ($p \in c(15, 45)$) frequently performed poorly across high treatment effect levels. Full detail in Table~\ref{tab:ext_worst_corr}.

%%% general overview and worst case %%%%%%%
\begin{figure}[h]
    \centering
    \includegraphics[width=1\linewidth]{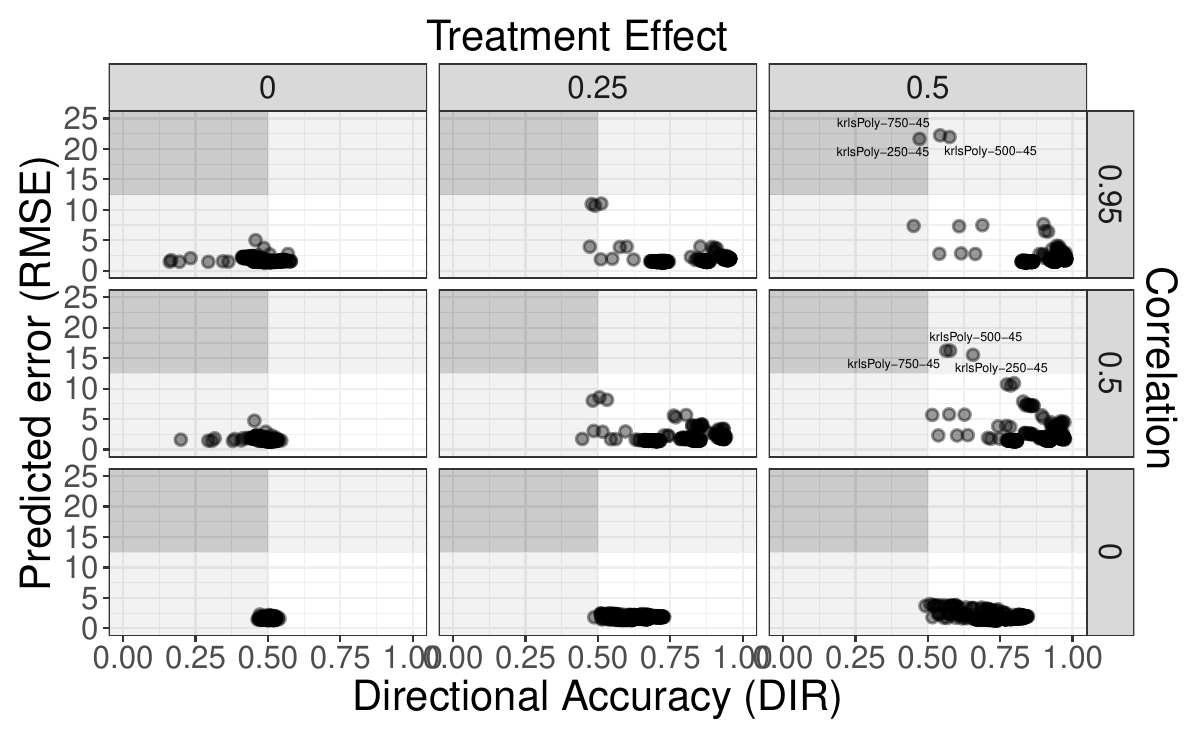}
    \caption{External validation (Correlated covariates): Relationship between predictive performance Root Mean Squared Error (RMSE) and Directional Accuracy (DIR) across models and treatment effect impact and correlation value between covariates. Each point represents a specific model-parameter configuration, labeled by model name, sample size (n), and number of predictors (p). Darkness gray areas reflect undesirable regions}
    \label{fig:Ext_risk_sensit_cor_te}
\end{figure}

\begin{table}[ht]
\centering
\begin{tabular}{rrlrrrrrrrr}
  \hline
$\mu_{\beta_\Delta}$ & $\rho$ & model & p & n & RMSE & R² & MAE & DIR & $\alpha$ & $\beta$ \\ 
  \hline
0.00 & 0.95 & Bayenet & 45 & 500 & 5.004 & 0.190 & 4.579 & 0.457 & 0.200 & 0.000 \\ 
  0.25 & 0.50 & krlsPoly & 45 & 250 & 8.580 & 0.620 & 7.003 & 0.506 & 0.000 & 0.000 \\ 
  0.25 & 0.50 & krlsPoly & 45 & 500 & 8.165 & 0.635 & 6.469 & 0.532 & 0.000 & 0.000 \\ 
  0.25 & 0.50 & rpart2 & 45 & 500 & 5.579 & 0.544 & 4.443 & 0.762 & 0.800 & 0.000 \\ 
  0.25 & 0.50 & krlsPoly & 45 & 750 & 8.028 & 0.659 & 6.382 & 0.483 & 0.000 & 0.000 \\ 
  0.25 & 0.50 & rpart2 & 45 & 750 & 5.288 & 0.575 & 4.224 & 0.769 & 0.800 & 0.000 \\ 
  0.25 & 0.95 & krlsPoly & 45 & 250 & 11.027 & 0.300 & 8.659 & 0.512 & 0.000 & 0.000 \\ 
  0.25 & 0.95 & krlsPoly & 45 & 500 & 10.671 & 0.457 & 8.595 & 0.491 & 0.000 & 0.000 \\ 
  0.25 & 0.95 & krlsPoly & 45 & 750 & 10.943 & 0.527 & 8.717 & 0.479 & 0.000 & 0.000 \\ 
  0.50 & 0.50 & krlsPoly & 15 & 250 & 5.676 & 0.452 & 4.555 & 0.515 & 0.000 & 0.000 \\ 
  0.50 & 0.50 & krlsPoly & 15 & 500 & 5.723 & 0.474 & 4.553 & 0.627 & 0.000 & 0.000 \\ 
  0.50 & 0.50 & krlsPoly & 15 & 750 & 5.773 & 0.497 & 4.603 & 0.574 & 0.000 & 0.000 \\ 
  0.50 & 0.50 & krlsPoly & 45 & 250 & 15.548 & 0.787 & 12.345 & 0.656 & 0.000 & 0.000 \\ 
  0.50 & 0.50 & rpart2 & 45 & 250 & 10.522 & 0.574 & 8.483 & 0.787 & 1.000 & 0.000 \\ 
  0.50 & 0.50 & krlsPoly & 45 & 500 & 16.248 & 0.812 & 13.088 & 0.577 & 0.000 & 0.000 \\ 
  0.50 & 0.50 & rpart2 & 45 & 500 & 10.902 & 0.568 & 8.620 & 0.798 & 0.800 & 0.000 \\ 
  0.50 & 0.50 & krlsPoly & 45 & 750 & 16.236 & 0.760 & 12.905 & 0.563 & 0.000 & 0.000 \\ 
  0.50 & 0.50 & rpart2 & 45 & 750 & 10.711 & 0.574 & 8.558 & 0.773 & 0.800 & 0.000 \\ 
  0.50 & 0.95 & krlsPoly & 15 & 250 & 7.357 & 0.410 & 5.827 & 0.451 & 0.000 & 0.000 \\ 
  0.50 & 0.95 & krlsPoly & 15 & 500 & 7.468 & 0.627 & 5.955 & 0.688 & 0.000 & 0.000 \\ 
  0.50 & 0.95 & krlsPoly & 15 & 750 & 7.295 & 0.665 & 5.823 & 0.608 & 0.000 & 0.000 \\ 
  0.50 & 0.95 & krlsPoly & 45 & 250 & 21.652 & 0.606 & 17.159 & 0.471 & 0.000 & 0.000 \\ 
  0.50 & 0.95 & krlsPoly & 45 & 500 & 21.960 & 0.602 & 17.486 & 0.575 & 0.000 & 0.000 \\ 
  0.50 & 0.95 & krlsPoly & 45 & 750 & 22.229 & 0.728 & 17.676 & 0.543 & 0.000 & 0.000 \\ 
   \hline
\end{tabular}
\caption{External validation (Correlated covariates): Scenarios where the Root Mean Squared Error (RMSE) exceeded 5 and Directional Accuracy (DIR) was below 0.8, suggesting suboptimal PITE estimation.}\label{tab:ext_worst_corr}
\end{table}

%%% best cases %%%%%%
\paragraph{Best-Performing Models Across External Validation Scenarios}

\begin{figure}[h]
    \centering
    \includegraphics[width=1\linewidth]{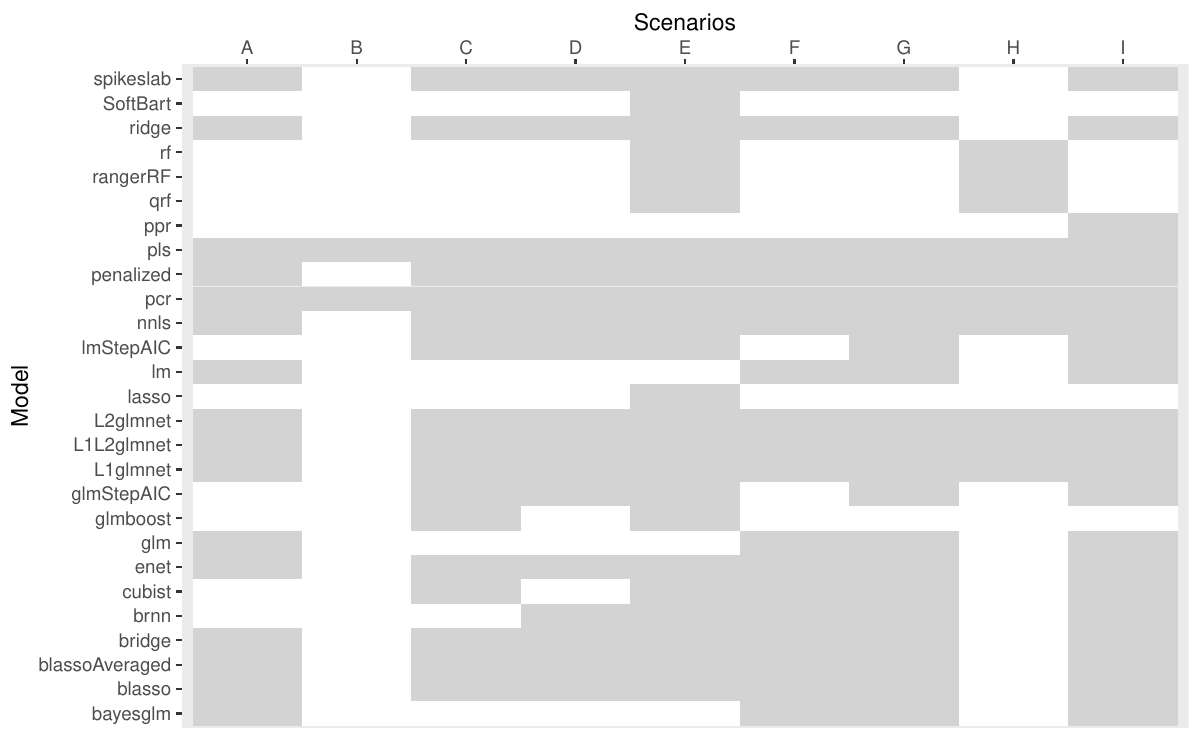}
    \caption{External validation (Correlated covariates): Models with lowest $\operatorname{RMSE}<2$  and High $\operatorname{DIR}>0.95$. Across Simulation Conditions [A--I] and dimension $p=45$.}\label{fig:best_risk_dir_external}
    \vspace{0.5em}
\parbox{\textwidth}{\small
\textit{setting:}
[A] $n=250$, $\rho=0.5$, ete=0.5; 
[B] $n=250$, $\rho=0.95$, ete=0.25; [C] $n=250$, $\rho=0.95$, ete=0.5; 
[D] $n=500$, $\rho=0.5$, ete=0.5; [E] $n=500$, $\rho=0.95$, ete=0.25; 
[F] $n=500$, $\rho=0.95$, ete=0.5; [G] $n=750$, $\rho=0.5$, ete=0.5; 
[H] $n=750$, $\rho=0.95$, ete=0.25; [I] $n=750$, $\rho=0.95$, ete=0.5.
}
\end{figure}

Figure~\ref{fig:best_risk_dir_external} summarizes the models that performed best under different simulation conditions in terms of minimizing the PITE RMSE while also maintaining directional accuracy---i.e., correctly identifying whether treatment is beneficial or harmful for each individual.

Overall, PLS and PCR consistently demonstrated robust performance across all nine simulation conditions. These models appeared particularly well-suited to settings with moderate to high multicollinearity ($\rho = 0.5$–$0.95$) and sample sizes of $n \in \{500, 750\}$, showing only limited degradation in the three scenarios with $n = 250$. They maintained low RMSE values and strong directional agreement across varying levels of the mean treatment effect coefficientes, $\mu_{\beta_{\Delta}}$. In contrast, more complex models such as lasso, ppr, SoftBart, and glmboost only succeeded under one or two conditions, suggesting that their performance may be more sensitive to the specific characteristics of the data-generating process. A group of models including ridge, spikeslab, bridge, and blasso showed intermediate robustness, performing well in approximately 7 out of the 9 settings. Collectively, these results suggest that simpler or regularized linear models often offer a favorable trade-off between predictive accuracy and interpretability for estimating PITE in structured simulation settings.

\subsubsection{Interactions between covariates}

The simulation study is designed to incorporate high-order interactions among covariates, reflecting scenarios in which treatment effects arise from complex underlying structures. Specifically, we generate samples of size $n \in \{500, 750, 1000\}$ with six baseline covariates, independently drawn from a multivariate normal distribution. From these covariates, we construct all possible interaction terms up to the sixth order, resulting in a full design matrix comprising 63 candidate predictors: the six main effects and all higher-order interactions through six-way terms. 

From this covariate set, a subset of size $p \in \{5, 15, 45\}$ is randomly selected to define the covariate set $W$, which is then used in the individual treatment bemefit generation.

\[
\Delta(X_i) = W_i^\top \beta_\Delta,
\]

Control potential outcomes are generated as a linear function dataset of baseline covariates $X$, also independently drawn from a multivariate normal distribution. Individual treatment effects are modeled as a linear function of $W$. The observed outcome for each individual is obtained by adding the treatment effect to the control outcome for treated individuals.

%\subsection{main results}
%%%% general overview %%%%%

under this scenario, model performance showed moderate accuracy with a mean RMSE of 1.77 (range: $1.42–4.64$). Directional accuracy (DIR), a key criterion for identifying beneficial versus harmful treatment, averaged 0.59, with a minimum of 0.46 and maximum of 0.73—indicating limited reliability in guiding individualized treatment decisions. 
 
\begin{figure}
    \centering
    \includegraphics[width=1\linewidth]{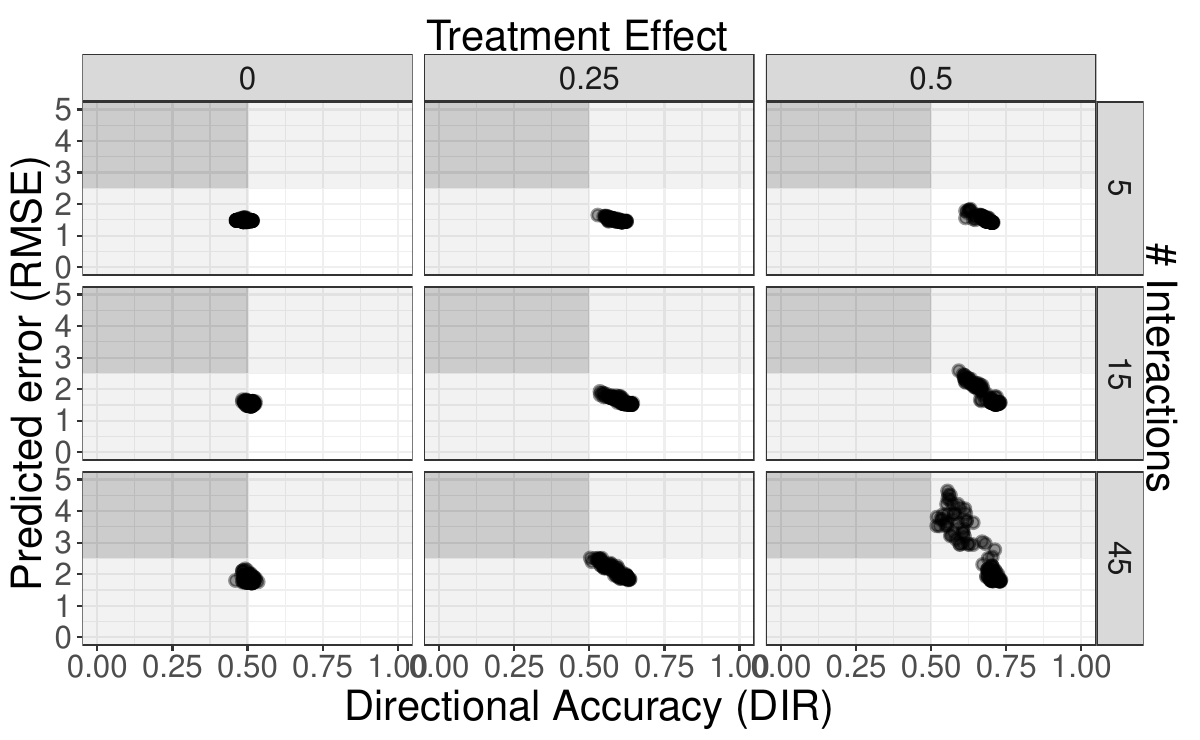}
    \caption{External validation (Interaction between covariates): Relationship between predictive performance Root Mean Squared Error (RMSE) and Directional Accuracy (DIR) across models and treatment effect impact and correlation value between covariates. Each point represents a specific model-parameter configuration, labeled by model name, sample size (n), and number of interactions. Darkness gray areas reflect undesirable regions}
    \label{fig:interac_risk_dir}
\end{figure}

Figure~\ref{fig:interac_risk_dir} illustrates the general variability in external validation including interaction effects based on how predictive accuracy and directional reliability vary as a function of both treatment effect size. Across all settings, models that fall into the upper-left grey region (high RMSE, low DIR) are of particular concern, as they not only predict poorly but also misclassify whether treatment is beneficial or harmful—posing a clinical risk if used for individualized decision-making. As the number of interaction terms increases, we observe greater variability in model performance, with more models drifting into the grey zone, especially when treatment effect is small or zero (e.g., leftmost panels).

\paragraph{Worst-Case Model Performance with Covariate Interactions}

%%%% worst cases %%%%%%%%%%

\begin{figure}
    \centering
    \includegraphics[width=1\linewidth]{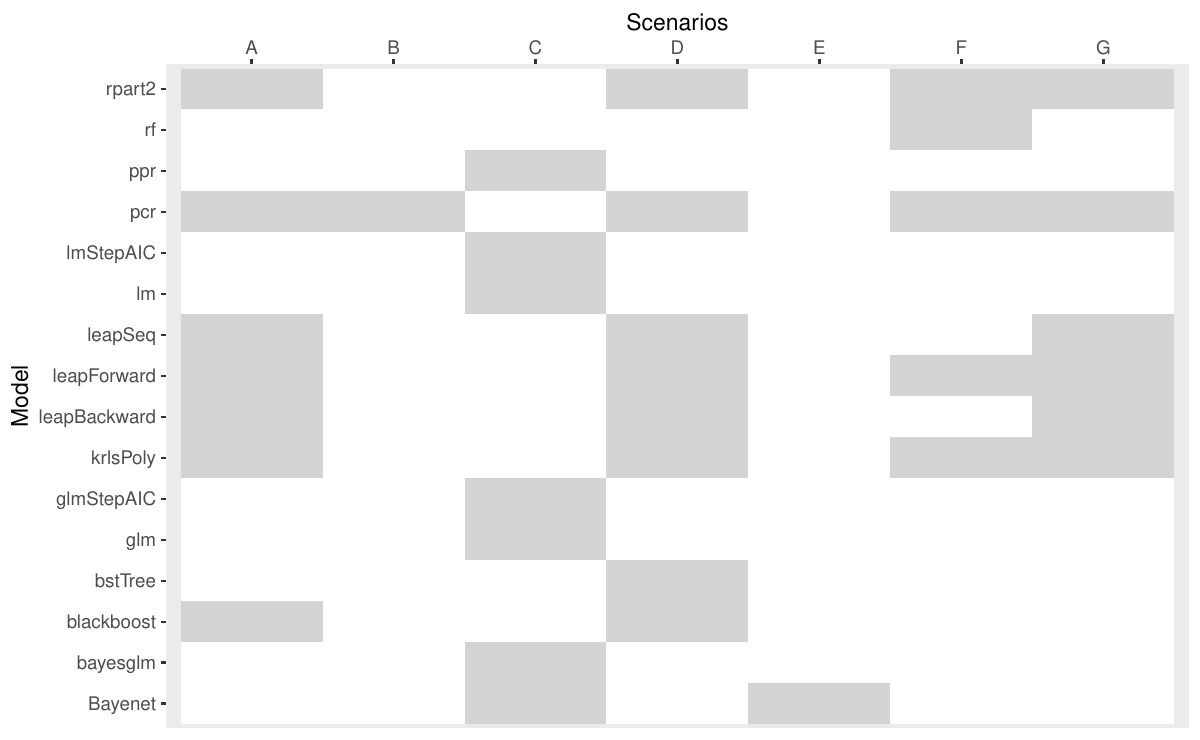}
    \caption{External validation (Interaction between covariates): Models with High $\operatorname{MSE}\geq 2$  and Low $\operatorname{DIR}\leq 0.55$. Across Simulation Conditions [A--G] and interactions dimensionality $45$.}
    \label{fig:interact_worst_risk_dir}
    \vspace{0.5em}
\parbox{\textwidth}{\small
\textit{setting:}
[A] $n=1000$, $\mu_{\beta_{\Delta}}$=0.25; 
[B] $n=1000$, $\mu_{\beta_{\Delta}}$=0.5; 
[C] $n=500$, $\mu_{\beta_{\Delta}}$=0; 
[D] $n=500$, $\mu_{\beta_{\Delta}}$=0.25; 
[E] $n=750$, $\mu_{\beta_{\Delta}}$=0; 
[F] $n=750$, $\mu_{\beta_{\Delta}}$=0.25; 
[G] $n=750$, $\mu_{\beta_{\Delta}}$=0.5; 
}
\end{figure}

Under external validation scenarios that included interaction effects among covariates and high dimensionality interactions ($45$), several models demonstrated poor performance, defined as having $\operatorname{RMSE} \geq 2$ and $\operatorname{DIR} \leq 0.55$. Figure~\ref{fig:interact_worst_risk_dir} shows that these worst-case failures occurred most frequently when the mean of the treatment effect, $\mu_{\beta_{\Delta}}$, was null or low, particularly under conditions [A] through [G], spanning sample sizes of 500 to 1000. Models such as pcr, rpart2, and krlsPoly exhibited the highest failure counts, failing in up to five conditions, while others such as leapForward, leapSeq, and blackboost failed in three or more. The consistent underperformance of these models in the presence of interaction effects suggests they are not well-suited for individualized treatment effect estimation under complex, high-dimensional settings.

%%%%%%%%%%%%% BEST case %%%%%%%%%%%%%%%
\paragraph{Best-Case Predictive Performance with Covariate Interactions}

Figure~\ref{fig:interac_best_risk_dir} illustrate that a wide range of models achieved high performance, defined as $\operatorname{RMSE} < 1.6$ and $\operatorname{DIR} > 0.65$, under external validation conditions with interaction effects and strong treatment signals ($\mu_{\beta_{\Delta}} = 0.5$). These successful cases were observed consistently across conditions [A] to [F], which involved moderate to large sample sizes ($n = 500–1000$) and lower interaction dimensionality ($\operatorname{number of interactions } = 5 \mbox{ or } \operatorname{number of interactions } 15$). Several models—including bayesglm, blasso, cubist, glmboost, ridge, spikeslab, and all glmnet variants—achieved perfect performance across all six conditions, demonstrating both high  predictive accuracy and reliable directional inference. Even more complex or adaptive models, such as BART, brnn, and bstTree maintained consistent performance, highlighting their robustness in interaction-rich scenarios.

\begin{figure}[h]
    \centering
    \includegraphics[width=1\linewidth]{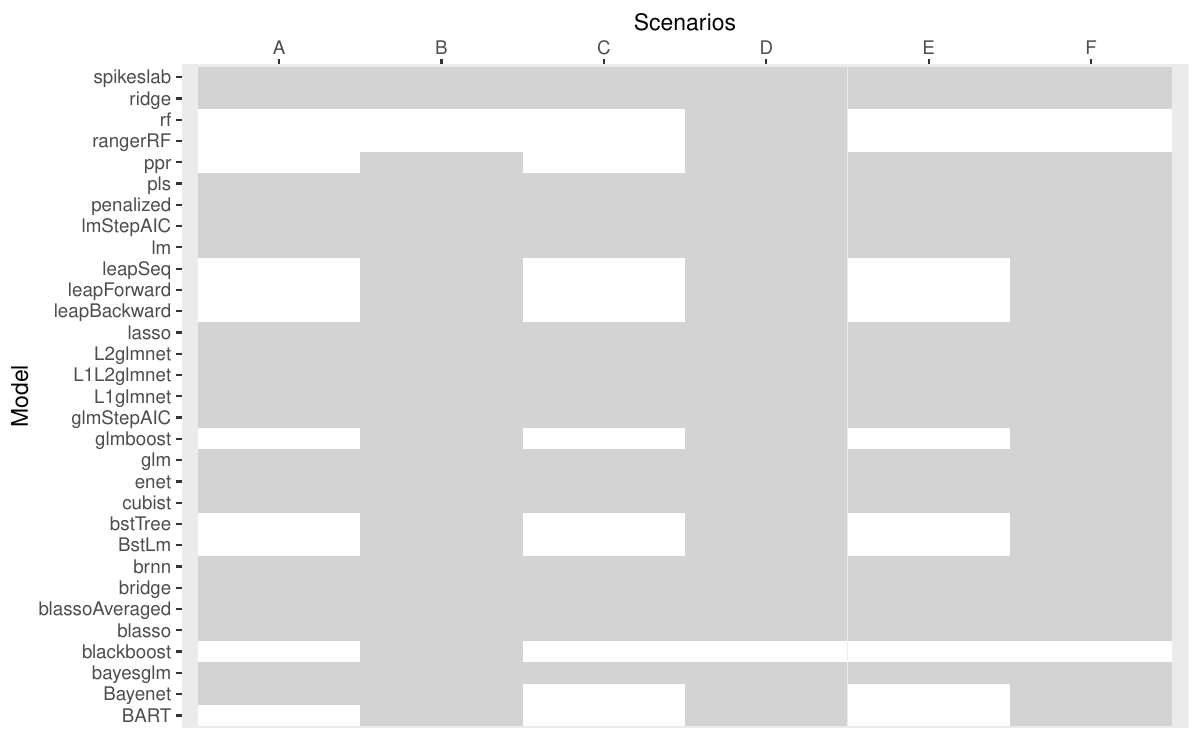}
    \caption{External validation (Interaction between covariates): Models Achieving High PITE Accuracy ($\operatorname{RMSE} < 1.6$ and $\operatorname{DIR} > 0.65$) Across External Validation Scenarios with Covariate Interactions}
    \label{fig:interac_best_risk_dir}
        \vspace{0.5em}
\parbox{\textwidth}{\small
\textit{setting:}
[A] $n=1000$, number of interactions $15$; 
[B] $n=1000$, number of interactions $5$; 
[C] $n=500$, number of interactions $15$; 
[D] $n=500$, number of interactions $5$; 
[E] $n=750$, number of interactions $15$; 
[F] $n=750$, number of interactions $5$.
}
\end{figure}

\subsection{Illustrative Example: Model-Specific Performance Across Individual Cases Complexity}
To examine how predictive accuracy varies across individual observations, we analyze a representative simulated dataset consisting of 250 validation patients based on the \emph{directional accuracy} indicator (Equation~\ref{eq:direction}), which quantifies the sign agreement between the estimated and true Predicted Individual Treatment Effects (PITE). For each patient $i$, we computed the proportion of models that correctly predicted the treatment direction,
\[
C_i = \frac{1}{K} \sum_{k=1}^{K} D_{i}^{(k)},
\]
where $K$ denotes the total number of evaluated models. The value $C_i \in [0,1]$ thus represents the fraction of methods correctly identifying whether the treatment benefited or harmed patient $i$, and serves as a measure of \emph{predictive consensus} or \emph{complexity}. Patients were stratified into twelve ordered complexity classes, ranging from $C_i=1$ (all models correct) to $C_i=0$ (no model correct), in 10\% increments.

Conceptually, the extremes of this spectrum are informative. Patients in the \textbf{100\% class} correspond to the simplest, most predictable cases, for whom all evaluated models correctly inferred the treatment effect direction—indicating clear, consistent treatment-response patterns. In contrast, the \textbf{0\% class} represents the most complex or ambiguous cases, where no model correctly predicted the treatment direction. These patients likely exhibit idiosyncratic treatment responses, weak signal-to-noise ratios, or interaction structures not well captured by any of the models. Intermediate classes (e.g., 40--70\%) reflect varying levels of predictive consensus, where only a subset of models correctly discerned the treatment direction.

For each class, we computed the proportion of correctly classified patients for every modeling approach. Results are summarized in Table~\ref{tab:directional_accuracy}. This complexity-based framework provides a flexible view of model robustness under increasing patient-level heterogeneity, facilitating direct comparison of predictive strategies in terms of both reliability and adaptability.

\begin{table}[!htbp]\footnotesize
\centering
\caption{Directional accuracy of 11 predictive modeling strategies across patients stratified by complexity.}
\label{tab:directional_accuracy}
\begin{tabular}{l r *{11}{r}}
\toprule
\textbf{class} & \textbf{number} & \textbf{pls} & \textbf{brnn} & \textbf{rf} & \textbf{bridge} & \textbf{enet} & \textbf{pcr} & \textbf{ridge} & \textbf{spikeslab} & \textbf{blasso} & \textbf{lasso} & \textbf{BART} \\
\midrule
0\% & 27 & 0.0 & 0.0 & 0.0 & 0.0 & 0.0 & 0.0 & 0.0 & 0.0 & 0.0 & 0.0 & 0.0 \\ 
  (0-10] & 8 & 0.0 & 12.5 & 25.0 & 0.0 & 0.0 & 25.0 & 12.5 & 0.0 & 0.0 & 0.0 & 25.0 \\ 
  (10-20] & 10 & 0.0 & 10.0 & 30.0 & 0.0 & 10.0 & 80.0 & 20.0 & 10.0 & 0.0 & 0.0 & 40.0 \\ 
  (20-30] & 5 & 20.0 & 20.0 & 60.0 & 20.0 & 0.0 & 60.0 & 40.0 & 40.0 & 0.0 & 0.0 & 40.0 \\ 
  (30-40] & 5 & 20.0 & 20.0 & 60.0 & 40.0 & 0.0 & 80.0 & 40.0 & 20.0 & 20.0 & 20.0 & 80.0 \\ 
  (40-50] & 8 & 12.5 & 62.5 & 50.0 & 62.5 & 37.5 & 50.0 & 62.5 & 12.5 & 50.0 & 37.5 & 62.5 \\ 
  (50-60] & 9 & 44.4 & 77.8 & 0.0 & 33.3 & 77.8 & 22.2 & 66.7 & 44.4 & 100.0 & 100.0 & 33.3 \\ 
  (60-70] & 10 & 90.0 & 100.0 & 40.0 & 60.0 & 80.0 & 20.0 & 70.0 & 70.0 & 70.0 & 80.0 & 20.0 \\ 
  (70-80] & 21 & 95.2 & 90.5 & 19.0 & 90.5 & 95.2 & 14.3 & 90.5 & 90.5 & 95.2 & 100.0 & 19.0 \\ 
  (80-90] & 21 & 95.2 & 95.2 & 33.3 & 100.0 & 100.0 & 38.1 & 90.5 & 95.2 & 95.2 & 95.2 & 61.9 \\ 
  (90-100) & 45 & 97.8 & 100.0 & 86.7 & 100.0 & 100.0 & 44.4 & 93.3 & 95.6 & 97.8 & 100.0 & 84.4 \\ 
  100\% & 81 & 100.0 & 100.0 & 100.0 & 100.0 & 100.0 & 100.0 & 100.0 & 100.0 & 100.0 & 100.0 & 100.0 \\ 
\bottomrule
\end{tabular}
\begin{tablenotes}
\small
\item \textbf{Notes:} The complexity classes reflect patient-level difficulty, defined by the fraction of methods that correctly predicted the true PITE direction. The number indicate the quantity of patients within each class for which each method correctly identified the treatment direction. Eleven widely used methods from representative families are included: penalized regression (lasso, blasso, enet, ridge), projection-based methods (pcr, pls), tree-based models (rf, BART), Bayesian approaches (bridge, spikeslab), and a neural network (brnn).
\end{tablenotes}
\end{table}

Table~\ref{tab:directional_accuracy} reports results from a representative simulated dataset consisting of 250 validation patients ($n = 500$, $p = 45$, $ete = 0.5$), with individuals stratified into 12 ordered complexity classes. Complexity was defined as the proportion of methods that correctly predicted the direction of the true Predicted Individual Treatment Effect (PITE) for each patient. 

Across the lowest-complexity stratum (90--100\%), all models exhibited near-perfect directional accuracy ($>95\%$). As complexity increased, however, performance diverged substantially. Flexible learners such as BART, BRNN, Ridge, and PCR remained robust even among moderately complex patients (30--60\%), whereas purely penalized linear models (Lasso, BLasso, Elastic Net) deteriorated rapidly, failing to recover the correct direction in high-complexity settings (challenging patients). Tree ensembles (RF) and Bayesian shrinkage methods (Bridge, Spike-and-Slab) displayed intermediate resilience, maintaining strong performance through the mid-range (40--70\%) but declining thereafter.

These results underscore the distinct operating regimes of predictive learners. Methods such as BART and BRNN combine nonlinear flexibility with strong regularization, enabling them to adapt gracefully as the signal structure becomes more intricate. Ridge and PCR exhibit robustness rooted in bias–variance stability, performing consistently across moderate levels of heterogeneity. In contrast, Lasso-type estimators, while interpretable and parsimonious, are highly sensitive to feature correlations and interaction effects, limiting their utility under complex treatment-response surfaces.

From a practical perspective, the \emph{complexity-based framework} presented here provides a structured diagnostic for assessing the stability of PITE estimates. The emergence of near-universal agreement at 100\% and complete disagreement at 0\% highlights the boundaries of current statistical learning methods: the former reflects predictably structured treatment responses, while the latter reveals fundamental uncertainty in PITE estimates. 

In applied contexts, these findings suggest that for populations dominated by low- to moderate-complexity patients, flexible Bayesian or ensemble learners (e.g., BART, BRNN) may offer the most reliable treatment direction estimates. When interpretability or feature selection is critical, penalized regressions can still serve as complementary tools within an ensemble or stacking framework. The interpretive lens of predictive consensus thus bridges methodological evaluation with clinical insight, guiding principled model choice for heterogeneous patient populations.

\section{Conclusion}

Our simulations confirm that internal validation---particularly when correlations are present but without interaction effects---offers a smoother, less volatile environment for evaluating predicted individual treatment effects (PITE). Under these conditions, many methods achieve both low estimation Root Mean Squared Error (RMSE) and high directional accuracy (DIR), masking potential weaknesses. In contrast, external validation exposes true model robustness by introducing distributional shifts through $\mu_{\beta_{\Delta}}$ and, in some scenarios, complex higher-order interactions. These interactions sharply reduce performance for several common approaches, underscoring the necessity of evaluating models under conditions that challenge their generality. External validation remains the gold standard for assessing whether a model’s internal performance can translate into clinical decision-making.

Across models and multiple validation regimes, regularized linear and projection-based methods---including ridge, lasso, elastic net, glmnet variants, partial least squares (PLS), and principal components regression (PCR)---consistently delivered low RMSE and high DIR across diverse sample sizes, multicollinearity levels, and treatment effect magnitudes. These approaches showed remarkable stability, avoiding the severe deterioration observed in kernel methods (e.g., krlsPoly) and recursive partitioning (rpart2), which often exhibited $\operatorname{RMSE} \geq 5$ and $\operatorname{DIR}<0.8$. Flexible learners such as Bayesian additive regression trees (BART), boosting algorithms, and certain Bayesian penalized regressions performed competitively, but only under favorable regimes---moderate-to-large samples, stronger treatment signals ($\mu_{\beta_{\Delta}} \approx 0.5$), and manageable dimensionality ($p \leq 15$). Considering strict thresholds ($\operatorname{RMSE} < 1, \operatorname{DIR} > 0.95$), even under internal validation, only a minority of models achieved repeated success across high-complexity conditions.

These results offer a roadmap for selecting and auditing models in precision medicine and clinical trial settings. First, penalized linear and projection-based methods provide robust defaults, maintaining stability across a wide range of realistic conditions. Second, flexible and nonparametric approaches—including tree-based ensembles and Bayesian learners—should be viewed as conditional strategies, effective when data characteristics align with their strengths but prone to instability otherwise. Third, both estimation RMSE and directionality are essential for assessing PITE validity under external validation, as a model that predicts effect magnitude well but misclassifies treatment direction can mislead clinical decisions. Finally, internal validation alone is insufficient for certifying a method’s clinical readiness; rigorous external validation is required to ensure safe and generalizable deployment. As the field moves toward Explainable AI (XAI), balancing predictive performance, interpretability, and robustness will be key to translating PITE estimation into actionable, trustworthy clinical tools.

\section{Discussion}
Our simulations show that in complex, high-dimensional settings, regularized linear and projection-based methods consistently provide the most reliable predicted individual treatment effects (PITE). In $p=45$ scenarios---mirroring modern biomedical applications such as genomics, imaging, and EHR-based care---PLS and PCR effectively reduced dimensionality while preserving predictive structure, yielding low RMSE and high DIR even at extreme multicollinearity ($\rho=0.95$). In contrast, models like krlsPoly and rpart2, frequently appeared in the poor-performing table, produced inflated RMSE (up to 22.2) and low DIR ($< 0.6$), especially when both correlation and treatment effect size were high—precisely where accurate treatment assignment is most critical.

External validation results reinforce the need for caution. Some models that performed well in-sample degraded sharply when applied to independent data, indicating poor generality. Regularized and projection-based approaches retained both magnitude and direction of treatment effects under these shifts, supporting their use when robustness is essential. Models insensitive to interactions---such as stepwise regression variants, PCR, and rpart2---were especially vulnerable, while lasso, ridge, elastic net, Bayesian approaches (blasso, bayesglm, spikeslab), and flexible learners (BART, brnn, cubist) achieved stable, accurate performance across varied conditions.

Internal validation alone proved insufficient to ensure stability. Poor performers clustered in the low-DIR, high-RMSE region even without interaction terms, underscoring the need for regularization and interaction-aware diagnostics. Although methods like glmnet, blasso, ridge, and PLS excelled internally, their strongest results occurred in settings with large treatment effect and sample size, highlighting the importance of stress-testing against weaker signals and external shifts.

These results, consistent with prior work \citep{Rolling2014}, emphasize that PITE-specific validation---using metrics like RMSE and DIR---should replace reliance on indirect outcome-based diagnostics. While our simulations captured high-dimensionality, multicollinearity, and complex interactions, future research should address calibration, domain adaptation, and real-world clinical datasets.

Robust, interpretable methods---particularly regularized or projection-based---offer a promising foundation for treatment stratification and personalized decision-making. By combining internal and external validation with directionality assessment, these approaches can deliver clinically actionable PITEs while avoiding overfitting and instability.

One important point to take into account is that, in our simulations, data-generating mechanisms follow linear regression assumptions, under which most models are expected to perform favorably. However, by evaluating PITE across linear, nonlinear, and tree-based methods, we explicitly quantify prediction differences between structural assumptions.

Finally, Table~\ref{tab:directional_accuracy} demonstrates that model performance within the PITE framework is strongly influenced by patient-level complexity. While most methods perform well in low-complexity (``easy'') cases, only a subset—most notably BART, BRNN, Ridge, and PCR—maintain high directional accuracy as complexity increases. These broadly robust learners exhibit near-perfect agreement in the 90–100\% class and retain meaningful predictive ability even under moderate heterogeneity, suggesting their resilience to noise and their ability to capture nonlinear or interaction effects. Flexible Bayesian and ensemble models such as Bridge, Spike-and-Slab, and Random Forest also perform competitively in low-to-intermediate complexity strata but show greater instability in highly complex cases, consistent with their sensitivity to weak or entangled treatment signals. In contrast, penalized linear methods including Lasso, BLasso, and Elastic Net, although valuable for interpretability and variable selection, exhibit sharp declines in performance as complexity increases, reflecting their limited flexibility in capturing higher-order dependencies. Projection-based methods like PLS and PCR perform adequately in low-complexity settings but deteriorate rapidly when the underlying latent structure misaligns with treatment-response mechanisms. Notably, across the most complex strata (0–10\%), no evaluated model succeeded in recovering the correct treatment direction, underscoring the inherent difficulty of these cases and the need for hybrid or ensemble frameworks that combine flexible learners with prior knowledge. Collectively, these results emphasize that model selection in PITE estimation should not rely solely on aggregate performance metrics but should account for the interaction between patient complexity and model adaptability. In practice, BART and BRNN emerge as reliable default strategies for precision medicine applications when extremely complex patients are rare, whereas combining penalized regression with more flexible learners may achieve an effective balance between interpretability and robustness in heterogeneous populations.

% \section*{Acknowledgments}
% Funding for this work was received from NIAAA grant R01AA030264 and the UK Medical Research Council grant MC\_UU\_00040/03. For the purpose of open access, the author has applied a Creative Commons Attribution (CC BY) license to any Author Accepted Manuscript version arising.

\section*{Declaration of Interest Statement}
No conflicts of interest were declared.

\bibliographystyle{plainnat}
\bibliography{pite_refs}

@Manual{Rsoftware,
  title        = {R: A Language and Environment for Statistical Computing},
  author       = {{R Core Team}},
  organization = {R Foundation for Statistical Computing},
  address      = {Vienna, Austria},
  year         = {2023},
  url          = {https://www.R-project.org/},
  note         = {R version 4.3.1 (2023-06-16, "Beagle Scouts")}
}

@Manual{zeallot,
  title = {zeallot: Multiple, Unpacking, and Destructuring Assignment},
  author = {Nathan Teetor},
  year = {2018},
  note = {R package version 0.1.0},
  url = {https://CRAN.R-project.org/package=zeallot},
}

@Manual{readxl,
  title = {readxl: Read Excel Files},
  author = {Hadley Wickham and Jennifer Bryan},
  year = {2023},
  note = {R package version 1.4.3},
  url = {https://CRAN.R-project.org/package=readxl},
}

@Article{MatchIt,
  title = {{MatchIt}: Nonparametric Preprocessing for Parametric Causal Inference},
  author = {Daniel E. Ho and Kosuke Imai and Gary King and Elizabeth A. Stuart},
  year = {2011},
  journal = {Journal of Statistical Software},
  volume = {42},
  number = {8},
  pages = {1--28},
  doi = {10.18637/jss.v042.i08},
}

@Manual{dplyr,
  title = {dplyr: A Grammar of Data Manipulation},
  author = {Hadley Wickham and Romain François and Lionel Henry and Kirill Müller and Davis Vaughan},
  year = {2023},
  note = {R package version 1.1.4},
  url = {https://CRAN.R-project.org/package=dplyr},
}

@Manual{tidyr,
  title = {tidyr: Tidy Messy Data},
  author = {Hadley Wickham and Davis Vaughan and Maximilian Girlich},
  year = {2024},
  note = {R package version 1.3.1},
  url = {https://CRAN.R-project.org/package=tidyr},
}

@Article{caret,
  title = {Building Predictive Models in R Using the caret Package},
  author = {{Kuhn} and {Max}},
  journal = {Journal of Statistical Software},
  year = {2008},
  volume = {28},
  number = {5},
  pages = {1--26},
  doi = {10.18637/jss.v028.i05},
  url = {https://www.jstatsoft.org/index.php/jss/article/view/v028i05},
}

@Manual{pls,
  title = {pls: Partial Least Squares and Principal Component Regression},
  author = {Kristian Hovde Liland and Bjørn-Helge Mevik and Ron Wehrens},
  year = {2024},
  note = {R package version 2.8-5},
  url = {https://CRAN.R-project.org/package=pls},
}

@Article{lme4,
  title = {Fitting Linear Mixed-Effects Models Using {lme4}},
  author = {Douglas Bates and Martin M{\"a}chler and Ben Bolker and Steve Walker},
  journal = {Journal of Statistical Software},
  year = {2015},
  volume = {67},
  number = {1},
  pages = {1--48},
  doi = {10.18637/jss.v067.i01},
}

@Manual{mboost,
  title = {mboost: Model-based Boosting in R},
  author = {Benjamin Hofner and Andreas Mayr and Nikola Robinzonov and Matthias Schmid},
  year = {2014},
  note = {R package},
}

@Manual{bst,
  title = {bst: Gradient Boosting},
  author = {Zhu Wang},
  year = {2022},
  note = {R package version 0.3-24},
  url = {https://CRAN.R-project.org/package=bst},
}

@Manual{monomvn,
  title = {monomvn: Estimation for MVN and Student-t Data with Monotone Missingness},
  author = {Robert B. Gramacy},
  year = {2024},
  note = {R package version 1.9-21},
  url = {https://CRAN.R-project.org/package=monomvn},
}

@Article{randomForest,
  title = {Classification and Regression by randomForest},
  author = {Andy Liaw and Matthew Wiener},
  journal = {R News},
  year = {2002},
  volume = {2},
  number = {3},
  pages = {18-22},
  url = {https://CRAN.R-project.org/doc/Rnews/},
}

@Article{party,
  title = {Unbiased Recursive Partitioning: A Conditional Inference Framework},
  author = {Torsten Hothorn and Kurt Hornik and Achim Zeileis},
  journal = {Journal of Computational and Graphical Statistics},
  year = {2006},
  volume = {15},
  number = {3},
  pages = {651--674},
  doi = {10.1198/106186006X133933},
}

@Manual{Cubist,
  title = {Cubist: Rule- And Instance-Based Regression Modeling},
  author = {Max Kuhn and Ross Quinlan},
  year = {2025},
  note = {R package version 0.5.0},
  url = {https://CRAN.R-project.org/package=Cubist},
}

@Manual{elasticnet,
  title = {elasticnet: Elastic-Net for Sparse Estimation and Sparse PCA},
  author = {Hui Zou and Trevor Hastie},
  year = {2020},
  note = {R package version 1.3},
  url = {https://CRAN.R-project.org/package=elasticnet},
}

@Manual{leaps,
  title = {leaps: Regression Subset Selection},
  author = {Thomas Lumley based on Fortran code by Alan Miller},
  year = {2024},
  note = {R package version 3.2},
  url = {https://CRAN.R-project.org/package=leaps},
}

@Manual{penalized,
  title = {penalized: L1 Penalized Estimation},
  author = {Jelle J. Goeman},
  year = {2022},
  note = {R package version 0.9-52},
}

@Article{KRLS,
  title = {Kernel-Based Regularized Least Squares in {R} ({KRLS}) and {Stata} ({krls})},
  author = {Jeremy Ferwerda and Jens Hainmueller and Chad J. Hazlett},
  journal = {Journal of Statistical Software},
  year = {2017},
  volume = {79},
  number = {3},
  pages = {1--26},
  doi = {10.18637/jss.v079.i03},
}

@Manual{quantregForest,
  title = {quantregForest: Quantile Regression Forests},
  author = {Nicolai Meinshausen},
  year = {2024},
  note = {R package version 1.3-7.1},
  url = {https://CRAN.R-project.org/package=quantregForest},
}

@Manual{spikeslab,
    title = {spikeslab : Prediction and Variable Selection Using Spike
      and Slab Regression},
    author = {H. Ishwaran and J.S. Rao and U.B. Kogalur},
    publisher = {manual},
    year = {2022},
    note = {R package version 1.1.6},
    url = {https://cran.r-project.org/package=spikeslab},
  }

@Manual{arm,
  title = {arm: Data Analysis Using Regression and Multilevel/Hierarchical Models},
  author = {Andrew Gelman and Yu-Sung Su},
  year = {2024},
  note = {R package version 1.14-4},
  url = {https://CRAN.R-project.org/package=arm},
}

@Manual{brnn,
  title = {brnn: Bayesian Regularization for Feed-Forward Neural Networks},
  author = {Paulino {Perez Rodriguez} and Daniel Gianola},
  year = {2025},
  note = {R package version 0.9.4},
  url = {https://CRAN.R-project.org/package=brnn},
}

@Article{ranger,
  title = {{ranger}: A Fast Implementation of Random Forests for High Dimensional Data in {C++} and {R}},
  author = {Marvin N. Wright and Andreas Ziegler},
  journal = {Journal of Statistical Software},
  year = {2017},
  volume = {77},
  number = {1},
  pages = {1--17},
  doi = {10.18637/jss.v077.i01},
}

@Manual{Bayenet,
  title = {Bayenet: Robust Bayesian Elastic Net},
  author = {Xi Lu and Cen Wu},
  year = {2025},
  note = {R package version 0.3},
  url = {https://CRAN.R-project.org/package=Bayenet},
}

@Article{SoftBart,
  title = {SoftBart: Soft Bayesian Additive Regression Trees},
  author = {Antonio R. Linero},
  journal = {arXiv e-prints},
  year = {2022},
}

@Article{BART,
  title = {Nonparametric Machine Learning and Efficient Computation with {B}ayesian Additive Regression Trees: The {BART} {R} Package},
  author = {Rodney Sparapani and Charles Spanbauer and Robert McCulloch},
  journal = {Journal of Statistical Software},
  year = {2021},
  volume = {97},
  number = {1},
  pages = {1--66},
  doi = {10.18637/jss.v097.i01},
}

@Article{glmnet,
  title = {Regularization Paths for Generalized Linear Models via Coordinate Descent},
  author = {Friedman J and Hastie T and Tibshirani R},
  journal = {Journal of Statistical Software},
  year = {2010},
  volume = {33},
  number = {1},
  pages = {1--22},
  doi = {10.18637/jss.v033.i01},
}

@Manual{glmnetUtils,
  title = {glmnetUtils: Utilities for 'Glmnet'},
  author = {Hong Ooi},
  year = {2023},
  note = {R package version 1.1.9},
  url = {https://CRAN.R-project.org/package=glmnetUtils},
}

@Book{ggplot2,
  author = {Hadley Wickham},
  title = {ggplot2: Elegant Graphics for Data Analysis},
  publisher = {Springer-Verlag New York},
  year = {2016},
  isbn = {978-3-319-24277-4},
  url = {https://ggplot2.tidyverse.org},
}

@Article{tidyverse,
  title = {Welcome to the {tidyverse}},
  author = {Hadley Wickham and Mara Averick and Jennifer Bryan and Winston Chang and Lucy D'Agostino McGowan and Romain François and Garrett Grolemund and Alex Hayes and Lionel Henry and Jim Hester and Max Kuhn and Thomas Lin Pedersen and Evan Miller and Stephan Milton Bache and Kirill Müller and Jeroen Ooms and David Robinson and Dana Paige Seidel and Vitalie Spinu and Kohske Takahashi and Davis Vaughan and Claus Wilke and Kara Woo and Hiroaki Yutani},
  year = {2019},
  journal = {Journal of Open Source Software},
  volume = {4},
  number = {43},
  pages = {1686},
  doi = {10.21105/joss.01686},
}

@misc{Zhao2017,
  doi = {10.48550/ARXIV.1705.08492},
  url = {https://arxiv.org/abs/1705.08492},
  author = {Zhao,  Yan and Fang,  Xiao and Simchi-Levi,  David},
  title = {Uplift Modeling with Multiple Treatments and General Response Types},
  publisher = {arXiv},
  year = {2017},
  copyright = {arXiv.org perpetual,  non-exclusive license}
}

@article{Rolling2014,
  title={Model selection for estimating treatment effects},
  author={Rolling, Craig A. and Yang, Yuhong},
  journal={Journal of the Royal Statistical Society: Series B (Statistical Methodology)},
  volume={76},
  number={4},
  pages={749--769},
  year={2014},
  publisher={Wiley Online Library}
}

@article{Efthimiou2023,
  title = {Measuring the performance of prediction models to personalize treatment choice},
  volume = {42},
  ISSN = {1097-0258},
  url = {http://dx.doi.org/10.1002/sim.9665},
  DOI = {10.1002/sim.9665},
  number = {8},
  journal = {Statistics in Medicine},
  publisher = {Wiley},
  author = {Efthimiou,  Orestis and Hoogland,  Jeroen and Debray,  Thomas P.A. and Seo,  Michael and Furukawa,  Toshiaki A. and Egger,  Matthias and White,  Ian R.},
  year = {2023},
  month = jan,
  pages = {1188–1206}
}

@article{Angrist2004,
 ISSN = {00130133, 14680297},
 URL = {http://www.jstor.org/stable/3590307},
 author = {Joshua D. Angrist},
 journal = {The Economic Journal},
 number = {494},
 pages = {C52--C83},
 publisher = {[Royal Economic Society, Wiley]},
 title = {Treatment Effect Heterogeneity in Theory and Practice},
 urldate = {2025-06-05},
 volume = {114},
 year = {2004}
}

@article{Hill2011,
  title = {Bayesian Nonparametric Modeling for Causal Inference},
  volume = {20},
  ISSN = {1537-2715},
  url = {http://dx.doi.org/10.1198/jcgs.2010.08162},
  DOI = {10.1198/jcgs.2010.08162},
  number = {1},
  journal = {Journal of Computational and Graphical Statistics},
  publisher = {Informa UK Limited},
  author = {Hill,  Jennifer L.},
  year = {2011},
  month = jan,
  pages = {217–240}
}

@article{Kent2018,
  title = {Personalized evidence based medicine: predictive approaches to heterogeneous treatment effects},
  ISSN = {1756-1833},
  url = {http://dx.doi.org/10.1136/bmj.k4245},
  DOI = {10.1136/bmj.k4245},
  journal = {BMJ},
  publisher = {BMJ},
  author = {Kent,  David M and Steyerberg,  Ewout and van Klaveren,  David},
  year = {2018},
  month = dec,
  pages = {k4245}
}

@book{Steyerberg2019,
  title = {Clinical Prediction Models: A Practical Approach to Development,  Validation,  and Updating},
  ISBN = {9783030163990},
  ISSN = {2197-5671},
  url = {http://dx.doi.org/10.1007/978-3-030-16399-0},
  DOI = {10.1007/978-3-030-16399-0},
  journal = {Statistics for Biology and Health},
  publisher = {Springer International Publishing},
  author = {Steyerberg,  Ewout W.},
  year = {2019}
}

@article{VanCalster2018,
  title = {Reporting and Interpreting Decision Curve Analysis: A Guide for Investigators},
  volume = {74},
  ISSN = {0302-2838},
  url = {http://dx.doi.org/10.1016/j.eururo.2018.08.038},
  DOI = {10.1016/j.eururo.2018.08.038},
  number = {6},
  journal = {European Urology},
  publisher = {Elsevier BV},
  author = {Van Calster,  Ben and Wynants,  Laure and Verbeek,  Jan F.M. and Verbakel,  Jan Y. and Christodoulou,  Evangelia and Vickers,  Andrew J. and Roobol,  Monique J. and Steyerberg,  Ewout W.},
  year = {2018},
  month = dec,
  pages = {796–804}
}

@article{Hoogland2024,
  title = {Evaluating individualized treatment effect predictions: A model‐based perspective on discrimination and calibration assessment},
  volume = {43},
  ISSN = {1097-0258},
  url = {http://dx.doi.org/10.1002/sim.10186},
  DOI = {10.1002/sim.10186},
  number = {23},
  journal = {Statistics in Medicine},
  publisher = {Wiley},
  author = {Hoogland,  J. and Efthimiou,  O. and Nguyen,  T. L. and Debray,  T. P. A.},
  year = {2024},
  month = aug,
  pages = {4481–4498}
}

@article{Gao2021,
  title = {Assessment of heterogeneous treatment effect estimation accuracy via matching},
  volume = {40},
  ISSN = {1097-0258},
  url = {http://dx.doi.org/10.1002/sim.9010},
  DOI = {10.1002/sim.9010},
  number = {17},
  journal = {Statistics in Medicine},
  publisher = {Wiley},
  author = {Gao,  Zijun and Hastie,  Trevor and Tibshirani,  Robert},
  year = {2021},
  month = apr,
  pages = {3990–4013}
}

@book{seber2004linear,
  title={Linear Regression Analysis},
  author={Seber, George AF and Lee, Alan J},
  year={2004},
  publisher={John Wiley \& Sons}
}

@book{mccullagh1989generalized,
  title={Generalized Linear Models},
  author={McCullagh, Peter and Nelder, John A},
  year={1989},
  publisher={Chapman and Hall/CRC}
}

@article{Tibshirani1996,
 ISSN = {00359246},
 URL = {http://www.jstor.org/stable/2346178},
 author = {Robert Tibshirani},
 journal = {Journal of the Royal Statistical Society. Series B (Methodological)},
 number = {1},
 pages = {267--288},
 publisher = {[Royal Statistical Society, Oxford University Press]},
 title = {Regression Shrinkage and Selection via the Lasso},
 urldate = {2025-06-03},
 volume = {58},
 year = {1996}
}

@article{Goeman2010,
  title = {L1Penalized Estimation in the Cox Proportional Hazards Model},
  volume = {52},
  ISSN = {1521-4036},
  url = {http://dx.doi.org/10.1002/bimj.200900028},
  DOI = {10.1002/bimj.200900028},
  number = {1},
  journal = {Biometrical Journal},
  publisher = {Wiley},
  author = {Goeman,  Jelle J.},
  year = {2010},
  month = feb,
  pages = {70–84}
}

@book{Lawson1995,
  title = {Solving Least Squares Problems},
  ISBN = {9781611971217},
  url = {http://dx.doi.org/10.1137/1.9781611971217},
  DOI = {10.1137/1.9781611971217},
  publisher = {Society for Industrial and Applied Mathematics},
  author = {Lawson,  Charles L. and Hanson,  Richard J.},
  year = {1995},
  month = jan 
}

@book{Jolliffe_PCA2002,
  ISBN = {0387954422},
  url = {http://dx.doi.org/10.1007/b98835},
  DOI = {10.1007/b98835},
  journal = {Springer Series in Statistics},
  publisher = {Springer-Verlag},
  title={Principal Component Analysis},
  author={Jolliffe, Ian T},
  year = {2002}
}

@article{Wold1987,
  title = {Principal component analysis},
  volume = {2},
  ISSN = {0169-7439},
  url = {http://dx.doi.org/10.1016/0169-7439(87)80084-9},
  DOI = {10.1016/0169-7439(87)80084-9},
  number = {1–3},
  journal = {Chemometrics and Intelligent Laboratory Systems},
  publisher = {Elsevier BV},
  author = {Wold,  Svante and Esbensen,  Kim and Geladi,  Paul},
  year = {1987},
  month = aug,
  pages = {37–52}
}

@article{Efron2004,
  title = {Least angle regression},
  volume = {32},
  ISSN = {0090-5364},
  url = {http://dx.doi.org/10.1214/009053604000000067},
  DOI = {10.1214/009053604000000067},
  number = {2},
  journal = {The Annals of Statistics},
  publisher = {Institute of Mathematical Statistics},
  author = {Efron,  Bradley and Hastie,  Trevor and Johnstone,  Iain and Tibshirani,  Robert},
  year = {2004},
  month = apr 
}

@article{Breiman2001,
title={Random Forests},
  volume = {45},
  ISSN = {0885-6125},
  url = {http://dx.doi.org/10.1023/A:1010933404324},
  DOI = {10.1023/a:1010933404324},
  number = {1},
  journal = {Machine Learning},
  publisher = {Springer Science and Business Media LLC},
  author = {Breiman,  Leo},
  year = {2001},
  pages = {5–32}
}

@article{meinshausen2006quantile,
  title={Quantile regression forests.},
  author={Meinshausen, Nicolai and Ridgeway, Greg},
  journal={Journal of machine learning research},
  volume={7},
  number={6},
  year={2006}
}

@Manual{Quinlan,
  title = {Cubist: Rule- And Instance-Based Regression Modeling},
  author = {Max Kuhn and Ross Quinlan},
  year = {2025},
  note = {R package version 0.5.0.9000, https://github.com/topepo/Cubist},
  url = {https://topepo.github.io/Cubist/},
}

@article{Brown1998,
 ISSN = {13697412, 14679868},
 URL = {http://www.jstor.org/stable/2985935},
 author = {P. J. Brown and M. Vannucci and T. Fearn},
 journal = {Journal of the Royal Statistical Society. Series B (Statistical Methodology)},
 number = {3},
 pages = {627--641},
 publisher = {[Royal Statistical Society, Wiley]},
 title = {Multivariate Bayesian Variable Selection and Prediction},
 urldate = {2025-06-03},
 volume = {60},
 year = {1998}
}

@article{Park2008,
  title = {The Bayesian Lasso},
  volume = {103},
  ISSN = {1537-274X},
  url = {http://dx.doi.org/10.1198/016214508000000337},
  DOI = {10.1198/016214508000000337},
  number = {482},
  journal = {Journal of the American Statistical Association},
  publisher = {Informa UK Limited},
  author = {Park,  Trevor and Casella,  George},
  year = {2008},
  month = jun,
  pages = {681–686}
}

@book{Gelman2006,
  title = {Data Analysis Using Regression and Multilevel/Hierarchical Models},
  ISBN = {9780511790942},
  url = {http://dx.doi.org/10.1017/CBO9780511790942},
  DOI = {10.1017/cbo9780511790942},
  publisher = {Cambridge University Press},
  author = {Gelman,  Andrew and Hill,  Jennifer},
  year = {2006},
  month = dec 
}

@article{tipping2001,
  author = {Tipping, Michael E.},
  journal = {Journal of Machine Learning Research},
  pages = {211--244},
  title = {{Sparse Bayesian Learning and the Relevance Vector Machine}},
  url = {http://jmlr.csail.mit.edu/papers/v1/tipping01a.html},
  volume = 1,
  year = 2001
}

@article{MacKay1992,
  title = {A Practical Bayesian Framework for Backpropagation Networks},
  volume = {4},
  ISSN = {1530-888X},
  url = {http://dx.doi.org/10.1162/neco.1992.4.3.448},
  DOI = {10.1162/neco.1992.4.3.448},
  number = {3},
  journal = {Neural Computation},
  publisher = {MIT Press},
  author = {MacKay,  David J. C.},
  year = {1992},
  month = may,
  pages = {448–472}
}

@article{Friedman1981,
  title = {Projection Pursuit Regression},
  volume = {76},
  ISSN = {1537-274X},
  url = {http://dx.doi.org/10.1080/01621459.1981.10477729},
  DOI = {10.1080/01621459.1981.10477729},
  number = {376},
  journal = {Journal of the American Statistical Association},
  publisher = {Informa UK Limited},
  author = {Friedman,  Jerome H. and Stuetzle,  Werner},
  year = {1981},
  month = dec,
  pages = {817–823}
}

@article{Bair2006,
  title = {Prediction by Supervised Principal Components},
  volume = {101},
  ISSN = {1537-274X},
  url = {http://dx.doi.org/10.1198/016214505000000628},
  DOI = {10.1198/016214505000000628},
  number = {473},
  journal = {Journal of the American Statistical Association},
  publisher = {Informa UK Limited},
  author = {Bair,  Eric and Hastie,  Trevor and Paul,  Debashis and Tibshirani,  Robert},
  year = {2006},
  month = mar,
  pages = {119–137}
}

@article{Nadaraya1964,
  title = {On Estimating Regression},
  volume = {9},
  ISSN = {1095-7219},
  url = {http://dx.doi.org/10.1137/1109020},
  DOI = {10.1137/1109020},
  number = {1},
  journal = {Theory of Probability \&amp; Its Applications},
  publisher = {Society for Industrial \& Applied Mathematics (SIAM)},
  author = {Nadaraya,  E. A.},
  year = {1964},
  month = jan,
  pages = {141–142}
}

@book{Koenker2005,
  title = {Quantile Regression},
  ISBN = {9780511754098},
  url = {http://dx.doi.org/10.1017/CBO9780511754098},
  DOI = {10.1017/cbo9780511754098},
  publisher = {Cambridge University Press},
  author = {Koenker,  Roger},
  year = {2005}
}

@article{Hastie-Zou2005,
 ISSN = {13697412, 14679868},
 URL = {http://www.jstor.org/stable/3647580},
 author = {Hui Zou and Trevor Hastie},
 journal = {Journal of the Royal Statistical Society. Series B (Statistical Methodology)},
 number = {2},
 pages = {301--320},
 publisher = {[Royal Statistical Society, Wiley]},
 title = {Regularization and Variable Selection via the Elastic Net},
 urldate = {2025-06-03},
 volume = {67},
 year = {2005}
}

@article{Hoerl2000,
  title = {Ridge Regression: Biased Estimation for Nonorthogonal Problems},
  volume = {42},
  ISSN = {1537-2723},
  url = {http://dx.doi.org/10.1080/00401706.2000.10485983},
  DOI = {10.1080/00401706.2000.10485983},
  number = {1},
  journal = {Technometrics},
  publisher = {Informa UK Limited},
  author = {Hoerl,  Arthur E. and Kennard,  Robert W.},
  year = {2000},
  month = feb,
  pages = {80–86}
}

@article{Jaki2024,
  title = {Predicting Individual Treatment Effects: Challenges and Opportunities for Machine Learning and Artificial Intelligence},
  ISSN = {1610-1987},
  url = {http://dx.doi.org/10.1007/s13218-023-00827-4},
  DOI = {10.1007/s13218-023-00827-4},
  journal = {KI - K\"{u}nstliche Intelligenz},
  publisher = {Springer Science and Business Media LLC},
  author = {Jaki,  Thomas and Chang,  Chi and Kuhlemeier,  Alena and Van Horn,  M. Lee},
  year = {2024},
  month = jan 
}

@article{Chang2021,
  doi = {10.1177/09622802211033640},
  url = {https://doi.org/10.1177/09622802211033640},
  year = {2021},
  month = sep,
  publisher = {{SAGE} Publications},
  volume = {30},
  number = {11},
  pages = {2369--2381},
  author = {Chi Chang and Thomas Jaki and Muhammad Saad Sadiq and Alena Kuhlemeier and Daniel Feaster and Natalie Cole and Andrea Lamont and Daniel Oberski and Yasin Desai and M. Lee Van Horn},
  title = {A permutation test for assessing the presence of individual differences in treatment effects},
  journal = {Statistical Methods in Medical Research}
}

@article{Dorie2016,
  title = {A flexible,  interpretable framework for assessing sensitivity to unmeasured confounding},
  volume = {35},
  ISSN = {1097-0258},
  url = {http://dx.doi.org/10.1002/sim.6973},
  DOI = {10.1002/sim.6973},
  number = {20},
  journal = {Statistics in Medicine},
  publisher = {Wiley},
  author = {Dorie,  Vincent and Harada,  Masataka and Carnegie,  Nicole Bohme and Hill,  Jennifer},
  year = {2016},
  month = may,
  pages = {3453–3470}
}

@article{Wager2018,
  title = {Estimation and Inference of Heterogeneous Treatment Effects using Random Forests},
  volume = {113},
  ISSN = {1537-274X},
  url = {http://dx.doi.org/10.1080/01621459.2017.1319839},
  DOI = {10.1080/01621459.2017.1319839},
  number = {523},
  journal = {Journal of the American Statistical Association},
  publisher = {Informa UK Limited},
  author = {Wager,  Stefan and Athey,  Susan},
  year = {2018},
  month = jun,
  pages = {1228–1242}
}

@article{Lamont2016,
  title = {Identification of predicted individual treatment effects in randomized clinical trials},
  volume = {27},
  ISSN = {1477-0334},
  url = {http://dx.doi.org/10.1177/0962280215623981},
  DOI = {10.1177/0962280215623981},
  number = {1},
  journal = {Statistical Methods in Medical Research},
  publisher = {SAGE Publications},
  author = {Lamont,  Andrea and Lyons,  Michael D and Jaki,  Thomas and Stuart,  Elizabeth and Feaster,  Daniel J and Tharmaratnam,  Kukatharmini and Oberski,  Daniel and Ishwaran,  Hemant and Wilson,  Dawn K and Van Horn,  M Lee},
  year = {2016},
  month = mar,
  pages = {142–157}
}

@ARTICLE{Kuhlemeier_2024,title={Validation of predicted individual treatment effects in out of sample respondents.},year={2024},author={Alena Kuhlemeier and Tomas Jaki and Katie Witkiewitz and Elizabeth A. Stuart and M Lee Van Horn},doi={10.1002/sim.10187},pmid={39075029},pmcid={null},mag_id={null},journal={Statistics in Medicine}}

@article{Ballarini2018,
  doi = {10.1371/journal.pone.0205971},
  url = {https://doi.org/10.1371/journal.pone.0205971},
  year = {2018},
  month = oct,
  publisher = {Public Library of Science ({PLoS})},
  volume = {13},
  number = {10},
  pages = {e0205971},
  author = {Nicol{\'{a}}s M. Ballarini and Gerd K. Rosenkranz and Thomas Jaki and Franz K\"{o}nig and Martin Posch},
  editor = {Alan Hubbard},
  title = {Subgroup identification in clinical trials via the predicted individual treatment effect},
  journal = {{PLOS} {ONE}}
}

@article{Friedman2001,
  title = {Greedy function approximation: A gradient boosting machine.},
  volume = {29},
  ISSN = {0090-5364},
  url = {http://dx.doi.org/10.1214/aos/1013203451},
  DOI = {10.1214/aos/1013203451},
  number = {5},
  journal = {The Annals of Statistics},
  publisher = {Institute of Mathematical Statistics},
  author = {Friedman,  Jerome H.},
  year = {2001},
  month = oct 
}

@article{Chipman2010,
  title = {BART: Bayesian additive regression trees},
  volume = {4},
  ISSN = {1932-6157},
  url = {http://dx.doi.org/10.1214/09-AOAS285},
  DOI = {10.1214/09-aoas285},
  number = {1},
  journal = {The Annals of Applied Statistics},
  publisher = {Institute of Mathematical Statistics},
  author = {Chipman,  Hugh A. and George,  Edward I. and McCulloch,  Robert E.},
  year = {2010},
  month = mar 
}

\newpage

\appendix
\section{Approaches}\label{app: xapproach}
\subsection*{Linear and Regularized Linear Models}

\textbf{Linear models}, the standard of statistical analysis, known for their interpretability. The \textbf{Linear Model (lm)} posits a direct linear relationship between predictors and the response, estimating coefficients by minimizing the sum of squared residuals \citep{seber2004linear}. While offering high interpretability and being computationally inexpensive, it is sensitive to linearity assumptions, normality of residuals, and multicollinearity. Expanding on this, the \textbf{Generalized Linear Model (glm)} provides increased flexibility by accommodating diverse outcome distributions (e.g., binomial for binary data, Poisson for count data) through various link functions \citep{mccullagh1989generalized}. This versatility, also achieved with low computational burden, makes it invaluable for non-normal dependent variables.

To address challenges inherent in complex datasets, particularly multicollinearity and high dimensionality, we then deployed \textbf{regularized linear models}. 
\textbf{Lasso Regression (lasso)} employs an L1 penalty, which promotes model sparsity by setting some coefficients precisely to zero, thereby performing intrinsic feature selection \citep{Tibshirani1996}. \textbf{Ridge Regression (ridge)} incorporates an L2 penalty, effectively shrinking coefficients to manage multicollinearity and reduce variance without forcing coefficients to zero \citep{Hoerl2000}. Both regularization are achieved with relatively low computational cost. Combining these strengths, \textbf{Elastic Net (enet)} leverages both L1 and L2 penalties, providing a robust solution that handles highly correlated predictors while still performing effective variable selection \citep{Hastie-Zou2005}; its computational demands remain low. For situations demanding highly tailored regularization, \textbf{Penalized Regression (penalized)} offers a flexible framework for custom penalty structures \citep{Goeman2010}, allowing specific control over regularization at minimal computational expense. Furthermore, when coefficients must inherently be non-negative, as in certain physical or economic applications, \textbf{Non-Negative Least Squares (nnls)} provides a constrained linear solution ensuring positivity \citep{Lawson1995}, also with low computational cost.

For datasets where dimensionality reduction or variable selection is paramount, we also utilized methods that transform or select features before modeling. \textbf{Principal Component Regression (pcr)} first orthogonally transforms predictors into principal components, then applies linear regression \citep{Jolliffe_PCA2002}. While effective for high-dimensional data with multicollinearity, this process incurs a medium computational cost and reduces direct interpretability. Similarly, \textbf{Partial Least Squares (pls)} constructs latent components that maximize covariance between predictors and response, making it highly effective for high-dimensional settings where prediction accuracy is prioritized \citep{Wold1987}. Its training also involves a medium computational cost. Finally, \textbf{Stepwise Regression (leapForward, leapBackward, lmStepAIC)}, an iterative variable selection technique guided by criteria like AIC, was used to identify salient predictors \citep{Efron2004}. This method typically presents a low computational cost, though its iterative nature can sometimes lead to local optima.

\subsection*{Tree-based and Ensemble Methods}

Our analysis also leveraged the power of \textbf{tree-based and ensemble methods}, known for their ability to capture complex nonlinear relationships and enhance predictive robustness. The foundational \textbf{Decision Tree (rpart2)} offers high interpretability through its rule-based structure and is computationally inexpensive, though it is susceptible to overfitting and high variance. To address this, \textbf{Random Forest (rf)}, an ensemble method, constructs multiple decision trees via bagging (bootstrap aggregating), significantly improving predictive accuracy and robustness against overfitting, particularly in high-dimensional and nonlinear contexts \citep{Breiman2001}. The training of Random Forest models typically involves a medium computational cost. An extension, \textbf{Quantile Random Forest (qrf)}, further enriches our understanding by predicting conditional quantiles of the response, providing detailed uncertainty quantification and insights into distributional tails \citep{meinshausen2006quantile}; this added detail also entails a medium computational cost.

Beyond bagging, \textbf{boosting techniques} were also implemented. The family of \textbf{Boosted Trees (blackboost, bstTree, gamboost, glmboost, BstLm)} sequentially builds weak learners to correct previous errors, consistently achieving high predictive performance \citep{Friedman2001}. These models generally incur a medium to high computational cost due to their iterative nature. Bridging ensemble methods with Bayesian principles, \textbf{Bayesian Trees (bartMachine)}, based on Bayesian Additive Regression Trees (BART), provides excellent predictive power while offering crucial uncertainty estimates for predictions through a sum of many simple trees \citep{Chipman2010}. The inherent complexity of its Bayesian inference makes its computational cost high. For models balancing interpretability with accuracy, \textbf{Rule-Based Tree (cubist)} generates interpretable rules, each associated with a local linear model, combining the strengths of decision rules and linear regression with a medium computational cost \citep{Quinlan}.

\subsection*{Bayesian Models}

To explicitly quantify uncertainty and incorporate prior knowledge, a suite of \textbf{Bayesian models} was employed. \textbf{Spike-and-Slab (spikeslab)} offers a principled Bayesian framework for simultaneous variable selection and parameter estimation by using specific priors that encourage sparsity \citep{Brown1998}. The reliance on Markov Chain Monte Carlo (MCMC) methods for inference makes its computational cost high, particularly for high-dimensional data. In a similar vein, the \textbf{Bayesian Lasso (blasso)} provides a probabilistic interpretation of Lasso regression with an L1 prior, delivering full posterior distributions for coefficients and inherently accounting for regularization parameter uncertainty \citep{Park2008}. Its computational cost is medium to high, depending on the complexity of the MCMC sampling. To enhance stability, \textbf{Averaged Bayesian Lasso (blassoAveraged)} reduces variability by averaging over multiple Bayesian Lasso models, which, as expected, translates to a high computational burden due to the multiple sampling runs.

Further expanding our Bayesian toolkit, \textbf{Bayesian Ridge Regression (bridge)} provides a probabilistic framework for coefficient shrinkage with an L2 prior, effectively handling multicollinearity while quantifying parameter uncertainty \citep{Hoerl2000, Goeman2010}. Its MCMC-based inference places its computational cost in the medium to high range. The \textbf{Bayesian GLM (bayesglm)} extends Generalized Linear Models to a Bayesian setting, allowing for probabilistic inference of parameters across various outcome distributions \citep{Gelman2006}. While highly flexible, its reliance on sampling methods generally incurs a medium computational cost. For nonlinear relationships, \textbf{Kernel Bayesian Regression (rvmPoly)}, a sparse Bayesian kernel method, learns a subset of "relevance vectors" critical for prediction, offering probabilistic outputs for complex nonlinear scenarios \citep{tipping2001}. This model also carries a medium to high computational cost.

\subsection*{Neural Networks}

For tasks demanding high flexibility and complex pattern recognition, we included \textbf{Neural Networks}. Specifically, the \textbf{Bayesian Neural Network (brnn)} integrates Bayesian inference by placing prior distributions over network weights \citep{MacKay1992}. This approach not only enhances robustness against overfitting but also provides invaluable uncertainty estimates for predictions, which is critical for making informed decisions in sensitive applications. The inherent complexity of Bayesian inference for neural networks results in a high computational cost.

\subsection*{Nonparametrics, kernel, latent variable, projections and constrained models}

Finally, several other specialized models were utilized for specific analytical advantages. \textbf{Projection Pursuit (ppr)} is a non-parametric method designed to uncover hidden nonlinear relationships by searching for "interesting" low-dimensional projections of predictors \citep{Friedman1981}. While powerful for exploratory analysis of high-dimensional data, it typically involves a medium computational cost. For supervised dimensionality reduction, \textbf{PCA-Based (superpc)} selects principal components based on their correlation with the response, subsequently used in a regression model \citep{Bair2006}. This method is particularly useful in high-dimensional domains like genomics and carries a medium computational cost. When nonlinear relationships are present but their form is unknown, \textbf{Kernel Regression (krlsPoly)} provides a flexible, non-parametric approach by weighting local data points \citep{Nadaraya1964}. This flexibility often comes with a medium computational cost, dependent on dataset size and kernel choice. Lastly, to understand how predictors influence different parts of the response distribution (not just the mean), \textbf{Quantile Regression (rqlasso)} was employed, which also incorporates Lasso regularization for variable selection at specific quantiles \citep{Koenker2005}. This approach, providing richer distributional insights, generally incurs a medium computational cost.

Latent Variable Models are about the unseen structure assumed to generate the data.
Constrained Models are about the rules or restrictions applied during the model fitting process.

\begin{table}[ht]
\centering
\caption{List of Models and Their Abbreviations}\label{app: abrev}
\begin{tabular}{ll}
\hline
\textbf{Full Name} & \textbf{Abbreviation} \\
\hline

Linear Regression & \texttt{lm} \\
Generalized Linear Model & \texttt{glm} \\
Ridge Regression & \texttt{ridge} \\
The Lasso & \texttt{lasso} \\
Elastic Net & \texttt{enet} \\
Penalized Linear Regression & \texttt{penalized} \\
Linear Model with Forward Selection & \texttt{leapForward} \\
Linear Model with Backward Selection & \texttt{leapBackward} \\
Linear Regression with Stepwise Selection & \texttt{lmStepAIC} \\
Principal Component Regression & \texttt{pcr} \\
Partial Least Squares & \texttt{pls} \\
Non-Negative Least Squares & \texttt{nnls} \\
Spike and Slab Regression & \texttt{spikeslab} \\
The Bayesian Lasso & \texttt{blasso} \\
Bayesian Ridge Regression (Model Averaged) & \texttt{blassoAveraged} \\
Bayesian Ridge Regression & \texttt{bridge} \\
Bayesian Generalized Linear Model & \texttt{bayesglm} \\
Bayesian Regularized Neural Networks & \texttt{brnn} \\
Random Forest & \texttt{rf} \\
Quantile Random Forest & \texttt{qrf} \\
CART & \texttt{rpart2} \\
Cubist & \texttt{cubist} \\
Boosted Tree & \texttt{blackboost} \\
Boosted Tree & \texttt{bstTree} \\
Boosted Generalized Additive Model & \texttt{gamboost} \\
Boosted Generalized Linear Model & \texttt{glmboost} \\
Boosted Linear Model & \texttt{BstLm} \\
Bayesian Additive Regression Trees & \texttt{bart-softbart} \\
Relevance Vector Machines with Polynomial Kernel & \texttt{rvmPoly} \\
Projection Pursuit Regression & \texttt{ppr} \\
Polynomial Kernel Regularized Least Squares & \texttt{krlsPoly} \\
Quantile Regression with LASSO Penalty & \texttt{rqlasso} \\
% Boosted Smoothing Spline & \texttt{bstSm} \\
% GLM with Stepwise Feature Selection & \texttt{glmStepAIC} \\
% Linear Model with Stepwise Selection & \texttt{leapSeq} \\
% Non-Convex Penalized Quantile Regression & \texttt{rqnc} \\
\hline
\end{tabular}
\label{tab:models}
\end{table}

\begin{landscape}
\begin{table}[h!]
    \centering
    \scriptsize % Using a very small font to ensure it fits
    \caption{Comprehensive Table of Regression Models}
    \label{tab:methods}
    \begin{tabular}{p{1.5cm} p{1.5cm} p{1.8cm} p{1.5cm} p{2cm} p{1.8cm} p{1.5cm}}
        \toprule
        \textbf{Model} & \textbf{Model Family} & \textbf{Regularization} & \textbf{Feature Selection} & \textbf{Interaction Handling} & \textbf{Interpretability} & \textbf{Comp. Cost} \\
        \midrule
        lm & Lin & None & No & Manual & High & Low \\
        glm & Lin & None & No & Manual & High & Low \\
        ridge & Lin & L2 & No & Manual & High & Low \\
        lasso & Lin & L1 & Yes & Manual & High & Low \\
        enet & Lin & L1 + L2 & Yes & Manual & High & Low \\
        penalized & Lin & L1/L2 & Yes & Manual & High & Low \\
        leapForward & Lin & None & Yes & Manual & High & Low \\
        leapBackward & Lin & None & Yes & Manual & High & Low \\
        lmStepAIC & Lin & None & Yes & Manual & High & Low \\
        pcr & Dim Reduct & None & No & Manual & Medium & Low \\
        pls & Latent Var & None & No & Manual & Medium & Low \\
        nnls & Constrained Lin & Non-negativity & No & Manual & High & Low \\
        spikeslab & Bayes & Prior & Yes & Manual & Medium & Medium \\
        blasso & Bayes & L1 Prior & Yes & Manual & Medium & Medium \\
        blassoAveraged & Bayes & L1 Prior & Yes & Manual & Medium & Medium \\
        bridge & Bayes & L2 Prior & No & Manual & Medium & Medium \\
        bayesglm & Bayes & Prior & No & Manual & Medium & Medium \\
        brnn & Neural Net & Prior & No & Automatic & Low & High \\
        rf & Ensemble/Tree & None & Yes & Automatic & Medium & Medium \\
        qrf & Ensemble/Tree & None & Yes & Automatic & Medium & Medium \\
        rpart2 & Tree & None & No & Automatic & High & Low \\
        cubist & Tree/Rule & None & No & Automatic & Medium & Medium \\
        blackboost & Boosted Trees & Shrinkage & No & Automatic & Medium & Medium \\
        bstTree & Boosted Trees & Shrinkage & No & Automatic & Medium & Medium \\
        gamboost & Boosted Additive & Shrinkage & Yes & Automatic & Medium & Medium \\
        glmboost & Boosted Lin & Shrinkage & Yes & Manual & Medium & Medium \\
        BstLm & Boosted Lin & Shrinkage & Yes & Manual & Medium & Medium \\
        bart & Bayes/Ensemble & Prior & Yes & Automatic & Medium & High \\
        rvmPoly & Kernel Method & Prior & Yes & Automatic & Low & High \\
        ppr & Projection & None & No & Automatic & Medium & Medium \\
        krlsPoly & Kernel Method & None & No & Automatic & Low & High \\
        rqlasso & Lin & L1 & Yes & Manual & Medium & Medium \\
        \bottomrule
    \end{tabular}
    \caption*{Notes: Lin (Linear), Dim Reduct (Dimensionality Reduction), Latent Var (Latent Variable), Bayes (Bayesian), Constrained Lin (Constrained Linear), Neural Net (Neural Network), Ensemble/Tree (Ensemble/Tree-Based), Tree/Rule (Tree-Based/Rule-Based), Boosted Trees (Boosted Tree-Based), Boosted Additive (Boosted Additive Models), Boosted Lin (Boosted Linear), Kernel Method (Kernel Method), Projection (Projection Pursuit).}
\end{table}
\end{landscape}

\section{PITE Metric Decompositions}\label{app: metrics}

\subsection{Proposition 1}\label{app: prop1}

\begin{proof}
By definition,
\[
\operatorname{Var}[\widehat{\operatorname{PITE}}(x)] = \operatorname{Var}[\hat{f}_1(x)] + \operatorname{Var}[\hat{f}_0(x)] - 2\, \operatorname{Cov}[\hat{f}_1(x), \hat{f}_0(x)].
\]
Since the MSE is the variance plus squared bias (bias-variance tradeoff), we write:
\[
\mathcal{R}_{\operatorname{PITE}} = \operatorname{Var}[\widehat{\operatorname{PITE}}(x)] + \operatorname{bias}_{\operatorname{PITE}}^2,
\]
where
\[
\operatorname{bias}_{\operatorname{PITE}} = \operatorname{bias}_1 - \operatorname{bias}_0.
\]
Expanding $\operatorname{bias}_{\operatorname{PITE}}^2$ and rearranging terms gives the stated result.
\end{proof}

We introduce the following notation to decompose PITE metrics and sampling variability:

\[
e_{1i} = f_{1}(x_i) - \hat{f}_{1}(x_i), \qquad 
e_{0i} = f_{0}(x_i) - \hat{f}_{0}(x_i),
\]

where $e_{1i}$ and $e_{0i}$ denote the estimation errors for the treatment and control predictive models, respectively. Similarly, define

\[
\epsilon_{1i} = f_{1}(x_i) - \bar{f}_{1}(x_i), \qquad 
\epsilon_{0i} = f_{0}(x_i) - \bar{f}_{0}(x_i),
\]

where $\bar{f}_{1}(x_i)$ and $\bar{f}_{0}(x_i)$ represent the expected predictive functions (e.g., under repeated sampling). Thus, $\epsilon_{1i}$ and $\epsilon_{0i}$ capture deviations from their expectation, isolating variability due to sampling.

\paragraph{PITE-$R^2$ Decomposition}
Consider
\[
R^2_{\text{PITE}} = 1 - \frac{\sum_{i=1}^n (PITE_i - \hat{PITE}_i)^2}{\sum_{i=1}^n (PITE_i - \bar{\text{PITE}})^2},
\]
and substituting \(PITE_i = f_{1i} - f_{0i}\) and \(\hat{PITE}_i = \hat{f}_{1i} - \hat{f}_{0i}\), the numerator becomes
\[
\sum_{i=1}^n (e_{1i} - e_{0i})^2 = \sum_{i=1}^n \big(e_{1i}^2 + e_{0i}^2 - 2 e_{1i} e_{0i} \big),
\]
and the denominator
\[
\sum_{i=1}^n (\epsilon_{1i} - \epsilon_{0i})^2 = \sum_{i=1}^n \big(\epsilon_{1i}^2 + \epsilon_{0i}^2 - 2 \epsilon_{1i} \epsilon_{0i} \big).
\]
Thus
\[
R^2_{\text{PITE}} = 1 - \frac{\mathrm{MSE}_1 + \mathrm{MSE}_0 - 2\,\mathrm{Cov}(e_1,e_0)}
{\mathrm{VAR}_1 + \mathrm{VAR}_0 - 2\,\mathrm{Cov}(\epsilon_1,\epsilon_0)},
\]
with $\mathrm{MSE}_t = \frac{1}{n}\sum e_{ti}^2$ and $\mathrm{VAR}_t = \frac{1}{n}\sum \epsilon_{ti}^2$ for $t \in \{0,1\}$.  
\begin{proof}

Assume that the Explained Variance (PITE-$R^2$) is given by:
$$
R^2_{PITE} = 1 - \frac{\sum_{i=1}^{n} (PITE_i - \hat{PITE}_i)^2}{\sum_{i=1}^{n} (PITE_i - \bar{PITE})^2}
$$
where  $PITE_i = f_{1i} - f_{0i}$ (true individual treatment effect),
      $\hat{PITE}_i = \hat{f}_{1i} - \hat{f}_{0i}$ (predicted treatment effect) and 
      $\bar{PITE} = \bar{f}_1 - \bar{f}_0$.

Using the identity $\hat{PITE}_i = \hat{f}_{1i} - \hat{f}_{0i}$ and $PITE_i = f_{1i} - f_{0i}$, we expand the numerator:

$$
(PITE_i - \hat{PITE}_i)^2 = \left[ (f_{1i} - \hat{f}_{1i}) - (f_{0i} - \hat{f}_{0i}) \right]^2 \\
= (e_{1i} - e_{0i})^2 = e_{1i}^2 + e_{0i}^2 - 2e_{1i}e_{0i},
$$

where $e_{1i} = f_{1i} - \hat{f}_{1i}$ and $e_{0i} = f_{0i} - \hat{f}_{0i}$. Therefore, the numerator becomes:

\begin{equation}
\sum_{i=1}^n (PITE_i - \hat{PITE}_i)^2 = \sum_{i=1}^n (e_{1i}^2 + e_{0i}^2 - 2e_{1i}e_{0i}).
\end{equation}

similar we work over 

$$
(PITE_i - \bar{PITE}_i)^2 = \left[ (f_{1i} - \bar{f}_{1i}) - (f_{0i} - \bar{f}_{0i}) \right]^2 \\
= (\epsilon_{1i} - \epsilon_{0i})^2 = \epsilon_{1i}^2 + \epsilon_{0i}^2 - 2\epsilon_{1i}\epsilon_{0i},
$$
where
      % $e_{1,i} = f_{1,i} - \hat{f}_{1,i}$
      % $e_{0,i} = f_{0,i} - \hat{f}_{0,i}$
      $\epsilon_{1,i} = f_{1,i} - \bar{f}_{1,i}$,
      $\epsilon_{0,i} = f_{0,i} - \bar{f}_{0,i}$

Finally, we obtain

$$
R^2_{PITE} = 1 - \frac{ \operatorname{MSE}_1 + \operatorname{MSE}_0 - 2\operatorname{Cov(e_1,e_0)}}{\operatorname{VAR}_1 + \operatorname{VAR}_0 - 2\operatorname{Cov(\epsilon_1,\epsilon_0)}}
$$

where $\operatorname{MSE}_1 = \frac{1}{n}\sum_{i=1}^{n}e_{1,i}^2$,
      $\operatorname{MSE}_0 = \frac{1}{n}\sum_{i=1}^{n}e_{0,i}^2$,
      $\operatorname{Cov}_e = \frac{1}{n}\sum_{i=1}^{n}e_{1,i}e_{0,i}$,
      $\operatorname{VAR}_1 = \frac{1}{n}\sum_{i=1}^{n}\epsilon_{1,i}^2$,
      $\operatorname{VAR}_0 = \frac{1}{n}\sum_{i=1}^{n}\epsilon_{0,i}^2$,
      $\operatorname{Cov}_f = \frac{1}{n}\sum_{i=1}^{n}\epsilon_{1,i}\epsilon_{0,i}$.
\end{proof}

Note that, PITE-$R^2$ depends not only on the individual outcome model $R^2$ values but also on $\mathrm{Cov}(e_1,e_0)$: positively correlated errors can improve $R^2_{\text{PITE}}$.

\paragraph{PITE-MAE Decomposition}
Based on
\[
\mathrm{MAE}_{\text{PITE}} = \frac{1}{n} \sum_{i=1}^n |e_{1i} - e_{0i}|,
\]
the triangle and reverse triangle inequalities give
\[
|\mathrm{MAE}_{f_1} - \mathrm{MAE}_{f_0}| \leq \mathrm{MAE}_{\text{PITE}} \leq \mathrm{MAE}_{f_1} + \mathrm{MAE}_{f_0},
\]
where $\mathrm{MAE}_{f_t} = \frac{1}{n} \sum_{i=1}^n |e_{ti}|$.  

\begin{proof}
\begin{equation}
MAE_{PITE} = \frac{1}{n} \sum_{i=1}^{n} \left| (f_{1i} - f_{0i}) - (\hat{f}_{1i} - \hat{f}_{0i}) \right| = \frac{1}{n} \sum_{i=1}^{n} |e_{1i} - e_{0i}|.
\end{equation}

{Upper Bound from triangle inequality}

$$
|e_{1i} - e_{0i}| \leq |e_{1i}| + |e_{0i}|,
$$

Lower bound from reverse triangle inequality:

$$
|e_{1i} - e_{0i}| \geq ||e_{1i}| - |e_{0i}||,
$$
Apply obth bounds across all observations:

$$
\frac{1}{n}\sum_{i=1}^n ||e_{1i}| - |e_{0i}|| \leq \frac{1}{n}\sum_{i=1}^n|e_{1i} - e_{0i}| \leq \frac{1}{n}\sum_{i=1}^n(|e_{1i}| + |e_{0i}|),
$$

resulting

$$
\frac{1}{n}\sum_{i=1}^n ||e_{1i}| - |e_{0i}|| \leq \operatorname{MAE}_{PITE} \leq \operatorname{MAE}_{f_1} + \operatorname{MAE}_{f_0},
$$

using Jensen and triangular inequality 
$$
||e_{1i}| - |e_{0i}|| \geq |\operatorname{MAE}_{f_1} - \operatorname{MAE}_{f_0}|
$$
we obtain 

$$
|\operatorname{MAE}_{f_1} - \operatorname{MAE}_{f_0}| \leq \operatorname{MAE}_{PITE} \leq \operatorname{MAE}_{f_1} + \operatorname{MAE}_{f_0},
$$

so

\begin{equation}
MAE_{PITE} \leq MAE_{f_1} + MAE_{f_0},
\end{equation}

where $MAE_{f_t} = \frac{1}{n} \sum_{i=1}^{n} |f_{ti} - \hat{f}_{ti}|$ for $t \in \{0,1\}$.
\end{proof}

This shows that the upper bound is tight when $e_1$ and $e_0$ have the same sign; the error can be much smaller if errors cancel.

\paragraph{PITE Calibration Decomposition}
Calibration regresses $PITE_i$ on $\hat{PITE}_i$:
\[
PITE_i = \alpha + \beta\,\hat{PITE}_i + \eta_i.
\]
If each outcome model satisfies
\[
f_{ti} = \alpha_t + \beta_t \hat{f}_{ti} + \varepsilon_{ti}, \quad t\in\{0,1\},
\]
then
\[
PITE_i = (\alpha_1 - \alpha_0) + \beta_1 \hat{f}_{1i} - \beta_0 \hat{f}_{0i} + (\varepsilon_{1i} - \varepsilon_{0i}).
\]
PITE calibration corresponds to $\alpha = \alpha_1 - \alpha_0$, $\beta_1 = \beta_0 = \beta$, and $\eta_i = \varepsilon_{1i} - \varepsilon_{0i}$.  
\begin{proof}
If both outcome models are well-calibrated:
\[
\begin{aligned}
f_{1i} &= \alpha_1 + \beta_1 \hat{f}_{1i} + \varepsilon_{1i}, \\
f_{0i} &= \alpha_0 + \beta_0 \hat{f}_{0i} + \varepsilon_{0i},
\end{aligned}
\]
then the treatment effect is:
\[
PITE_i = f_{1i} - f_{0i} = (\alpha_1 - \alpha_0) + \beta_1 \hat{f}_{1i} - \beta_0 \hat{f}_{0i} + (\varepsilon_{1i} - \varepsilon_{0i}).
\]

For this to be equivalent to:
\[
PITE_i = \alpha + \beta \hat{PITE}_i + \eta_i,
\]
we require:
\[
\boxed{
\alpha = \alpha_1 - \alpha_0, \quad \beta_1 = \beta_0 = \beta, \quad \eta_i = \varepsilon_{1i} - \varepsilon_{0i}.
}
\]
\end{proof}

Note that perfect calibration in $f_1$ and $f_0$ does not ensure PITE calibration unless slopes match and intercepts differ appropriately.

Some relationship used. Let
\[
e_{1i} = f_1{(x_i)} - \hat{f}_1{{(x_i)}}, \quad e_{0i} = f_0{{(x_i)}} - \hat{f}_0{(x_i)}, \quad 
\epsilon_{1i} = f_1{(x_i)} - \bar{f}_1{(x_i)}, \quad \epsilon_0{(x_i)} = f_0{(x_i)} - \bar{f}_0{(x_i)}.
\]

\paragraph{Empirical adequate diagnostic metrics evaluation}
%%%%%%%%%%%%%%%%%%%%%%%%%%%%%%%%%%%%%%%%%
%%% interaction and External validation
%%%%%%%%%%%%%%%%%%%%%%%%%%%%%%%%%%%%%%%%%

To analyze the main contribution of the diagnostic metrics a principal component analysis explain 67\% of the variability in two components showing that PC1 (standard deviation = 1.67) is predominantly influenced by Rsquared, DIR, and $\alpha$, reflecting their strong contribution to the main variation (variability) in the diagnostic metrics information. PC2 (standard deviation = 1.12) is primarily defined by RMSE and MAE, indicating a dimension related to error metrics.

\begin{figure}[h]
    \centering
    \includegraphics[width=1\linewidth]{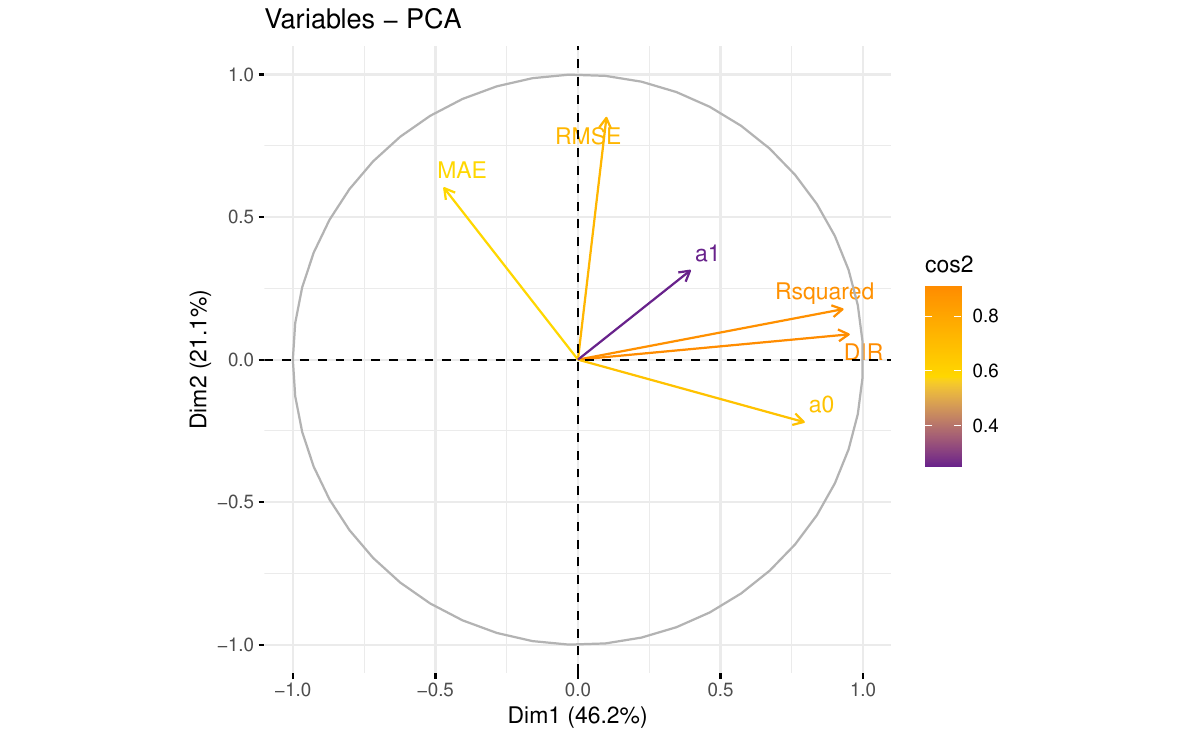}\\
    \includegraphics[width=1\linewidth]{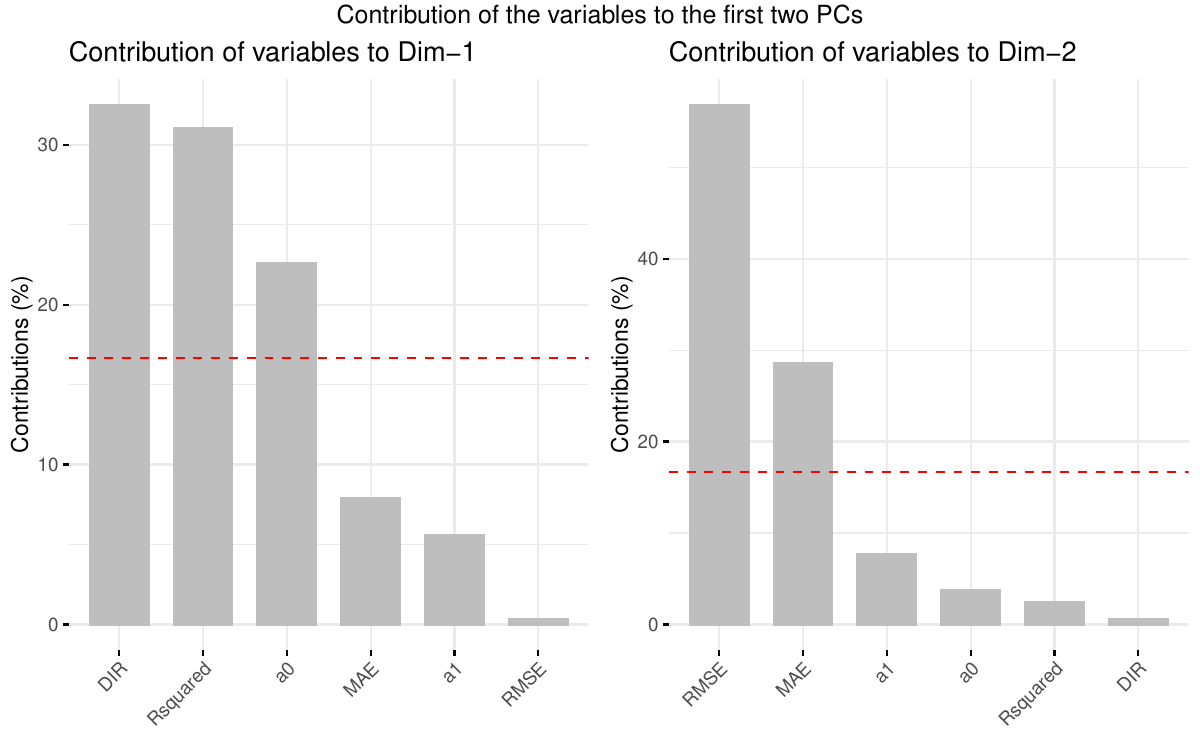}
    \caption{Enter Caption}
    \label{fig:enter-label}
\end{figure}

\section{Simulation}\label{app:sim_details}
\subsection*{Software and Model Implementation}

Analyses were conducted in R version 4.3.1~\cite{Rsoftware} on a Windows platform (x86\_64-w64-mingw32, ucrt). R packages used for data manipulation, modeling, and visualization, including \texttt{zeallot}~\cite{zeallot}, \texttt{readxl}~\cite{readxl}, \texttt{MatchIt}~\cite{MatchIt}, \texttt{dplyr}~\cite{dplyr}, \texttt{tidyr}~\cite{tidyr}, \texttt{caret}~\cite{caret}, \texttt{pls}~\cite{pls}, \texttt{lme4}~\cite{lme4}, \texttt{mboost}~\cite{mboost}, \texttt{bst}~\cite{bst}, \texttt{monomvn}~\cite{monomvn}, \texttt{randomForest}~\cite{randomForest}, \texttt{party}~\cite{party}, \texttt{Cubist}~\cite{Cubist}, \texttt{elasticnet}~\cite{elasticnet}, \texttt{leaps}~\cite{leaps}, \texttt{penalized}~\cite{penalized}, \texttt{KRLS}~\cite{KRLS}, \texttt{quantregForest}~\cite{quantregForest}, \texttt{spikeslab}~\cite{spikeslab}, \texttt{arm}~\cite{arm}, \texttt{brnn}~\cite{brnn}, \texttt{ranger}~\cite{ranger}, \texttt{Bayenet}~\cite{Bayenet}, \texttt{BART}~\cite{BART}, \texttt{glmnet}~\cite{glmnet}, \texttt{glmnetUtils}~\cite{glmnetUtils}, \texttt{ggplot2}~\cite{ggplot2}, and \texttt{tidyverse}~\cite{tidyverse}.

\paragraph{General Training Procedure:}  
Models were trained separately on treated and control subsets. For methods implemented via \texttt{caret::train}, we employed 10-fold cross-validation (\texttt{trainControl(method = "cv", number = 10)}), and predictions on the test set were obtained using \texttt{predict(fit, newdata=test)}. 
\paragraph{Bayesian Additive Regression Trees (BART):}  
BART models were fit using \texttt{wbart}, with 50 burn-in iterations (\texttt{nskip = 50}) and 200 posterior draws (\texttt{ndpost = 200}). Posterior predictive means (\texttt{yhat.test.mean}) were used for outcome predictions, allowing flexible modeling of nonlinearities and interactions.

\paragraph{Random Forest (RF):}  
Random Forest models were implemented using \texttt{ranger} with 500 trees and 3 parallel threads (\texttt{num.threads = 3L}). Predictions were computed as the average over all trees.

\paragraph{Bayesian Lasso Network (BayeNet):}  
BayeNet models were fit with \texttt{max.steps = 500} and \texttt{penalty = "lasso"}. Predictions were obtained using the posterior predictive mean via \texttt{predict(fit, X, clin)}.

\paragraph{Elastic Net and Variants (\texttt{glmnet}):}  
Models were fit using \texttt{cva.glmnet} with cross-validated selection of the regularization parameter (\(\lambda\)). L1 (lasso, \texttt{alpha = 1.0}), L2 (ridge, \texttt{alpha = 0}), and mixed (elastic net, \texttt{alpha = 0.2}) penalties were used. Predictions were obtained via \texttt{predict(fit, newx)}.

\subsection{External-correlations}

RMSE ranged from 1.34 to 22.23 (mean = 1.84, median = 1.56), indicating generally low error with occasional large outliers. DIR spanned from 0.16 to 0.98 (mean = 0.69, median = 0.70), suggesting good directional performance in most cases. R² values ranged from -0.007 to 0.995 (mean = 0.43, median = 0.35), confirming strong performance in some settings but poor fit in others. MAE ranged from 1.05 to 17.68 with a median of 1.25 . Calibration metrics were also dispersed: $a_0$ ranged from 0 to 1 (median = 0.60), and $a_1$ from 0 to 0.8 (median = 0), indicating common underestimation of treatment effects. These patterns confirm substantial variability and the presence of outliers across metrics.

% \begin{verbatim}
% RMSE             adjRsquared            MAE
% Min.   : 1.337   Min.   :-0.00667   Min.   : 1.046
% 1st Qu.: 1.468   1st Qu.: 0.03796   1st Qu.: 1.175
% Median : 1.564   Median : 0.35256   Median : 1.247
% Mean   : 1.835   Mean   : 0.42946   Mean   : 1.463
% 3rd Qu.: 1.870   3rd Qu.: 0.76893   3rd Qu.: 1.495
% Max.   :22.229   Max.   : 0.99507   Max.   :17.676
%                  NA s   :14
%       DIR               a               b
%  Min.   :0.1610   Min.   :0.0000   Min.   :0.00000
%  1st Qu.:0.5103   1st Qu.:0.4000   1st Qu.:0.00000
%  Median :0.6983   Median :0.6000   Median :0.00000
%  Mean   :0.6941   Mean   :0.5976   Mean   :0.04766
%  3rd Qu.:0.8550   3rd Qu.:0.8000   3rd Qu.:0.00000
%  Max.   :0.9783   Max.   :1.0000   Max.   :0.80000

% \end{verbatim}

\begin{table}[h!]
    \centering
    \caption{Summary statistics for model performance metrics.}
    \label{tab:summary_stats}
    \begin{tabular}{l|c|c|c|c|c|c}
        \hline
        & \textbf{RMSE} & \textbf{adjRsquared} & \textbf{MAE} & \textbf{DIR} & $\alpha$ & $\beta$\\
        \hline
        \textbf{Min.} & 1.337 & -0.00667 & 1.046 & 0.1610 & 0.0000 & 0.00000 \\
        \textbf{1st Qu.} & 1.468 & 0.03796 & 1.175 & 0.5103 & 0.4000 & 0.00000 \\
        \textbf{Median} & 1.564 & 0.35256 & 1.247 & 0.6983 & 0.6000 & 0.00000 \\
        \textbf{Mean} & 1.835 & 0.42946 & 1.463 & 0.6941 & 0.5976 & 0.04766 \\
        \textbf{3rd Qu.} & 1.870 & 0.76893 & 1.495 & 0.8550 & 0.8000 & 0.00000 \\
        \textbf{Max.} & 22.229 & 0.99507 & 17.676 & 0.9783 & 1.0000 & 0.80000 \\
        \textbf{NA's} & & 14 & & & & \\
        \hline
    \end{tabular}
\end{table}

\subsection{External-interactions}

MAE of 1.37. Adjusted R² values were low overall (mean = 0.22), and in some cases negative, reflecting poor model generalization. Moreover, the personalization component was weak, with negligible $a_1$ values and mean $a_0=0.17$, suggesting limited treatment effect heterogeneity capture.

% \begin{verbatim}
%      RMSE        adjRsquared             MAE             DIR        
%  Min.   :1.416   Min.   :-0.003479   Min.   :1.134   Min.   :0.4600  
%  1st Qu.:1.498   1st Qu.: 0.002966   1st Qu.:1.199   1st Qu.:0.5108  
%  Median :1.582   Median : 0.147590   Median :1.261   Median :0.5916  
%  Mean   :1.774   Mean   : 0.221663   Mean   :1.374   Mean   :0.5901  
%  3rd Qu.:1.851   3rd Qu.: 0.370098   3rd Qu.:1.461   3rd Qu.:0.6612  
%  Max.   :4.640   Max.   : 0.756603   Max.   :2.697   Max.   :0.7324  
%                  NA's   :1                                           
%        a               b   
%  Min.   :0.0000   Min.   :0  
%  1st Qu.:0.0000   1st Qu.:0  
%  Median :0.0000   Median :0  
%  Mean   :0.1688   Mean   :0  
%  3rd Qu.:0.2000   3rd Qu.:0  
%  Max.   :1.0000   Max.   :0      
% \end{verbatim} 

\begin{table}[h!]
    \centering
    \caption{Summary statistics for model performance metrics.}
    \label{tab:new_summary_stats}
    \begin{tabular}{l|c|c|c|c|c|c}
        \hline
        & \textbf{RMSE} & \textbf{adjRsquared} & \textbf{MAE} & \textbf{DIR} & $\alpha$ & $\beta$ \\
        \hline
        \textbf{Min.} & 1.416 & -0.003479 & 1.134 & 0.4600 & 0.0000 & 0 \\
        \textbf{1st Qu.} & 1.498 & 0.002966 & 1.199 & 0.5108 & 0.0000 & 0 \\
        \textbf{Median} & 1.582 & 0.147590 & 1.261 & 0.5916 & 0.0000 & 0 \\
        \textbf{Mean} & 1.774 & 0.221663 & 1.374 & 0.5901 & 0.1688 & 0 \\
        \textbf{3rd Qu.} & 1.851 & 0.370098 & 1.461 & 0.6612 & 0.2000 & 0 \\
        \textbf{Max.} & 4.640 & 0.756603 & 2.697 & 0.7324 & 1.0000 & 0 \\
        \textbf{NA's} & & 1 & & & & \\
        \hline
    \end{tabular}
\end{table}

\subsection{Internal correlated}
The median adjusted R² of 0.81 and the low median MAE (0.40) indicate that most models accurately captured individual treatment effects under these controlled, internal conditions. Moreover, the moderate average values of $\alpha = 0.33$ and $\beta = 0.22$ reflect partial personalization, suggesting that some models leveraged patient-level heterogeneity in treatment assignment.

% \begin{verbatim}

%       RMSE           adjRsquared             DIR              MAE         
%  Min.   : 0.02567   Min.   :-0.006055   Min.   :0.0224   Min.   : 0.0206  
%  1st Qu.: 0.32037   1st Qu.: 0.200436   1st Qu.:0.5376   1st Qu.: 0.2590  
%  Median : 0.50757   Median : 0.806061   Median :0.8448   Median : 0.4048  
%  Mean   : 0.87941   Mean   : 0.622523   Mean   :0.7635   Mean   : 0.6974  
%  3rd Qu.: 0.87909   3rd Qu.: 0.980559   3rd Qu.:0.9560   3rd Qu.: 0.6934  
%  Max.   :21.99891   Max.   : 0.999999   Max.   :1.0000   Max.   :17.5477  
%                     NA's   :16                                            
%        a               b        
%  Min.   :0.0000   Min.   :0.0000  
%  1st Qu.:0.0000   1st Qu.:0.0000  
%  Median :0.2000   Median :0.2000  
%  Mean   :0.3281   Mean   :0.2175  
%  3rd Qu.:0.6000   3rd Qu.:0.4000  
%  Max.   :1.0000   Max.   :1.0000  
     
% \end{verbatim}

\begin{table}[h!]
    \centering
    \caption{Summary statistics for model performance metrics.}
    \label{tab:latest_summary_stats}
    \begin{tabular}{l|c|c|c|c|c|c}
        \hline
        & \textbf{RMSE} & \textbf{adjRsquared} & \textbf{DIR} & \textbf{MAE} & $\alpha$ & $\beta$ \\
        \hline
        \textbf{Min.} & 0.02567 & -0.006055 & 0.0224 & 0.0206 & 0.0000 & 0.0000 \\
        \textbf{1st Qu.} & 0.32037 & 0.200436 & 0.5376 & 0.2590 & 0.0000 & 0.0000 \\
        \textbf{Median} & 0.50757 & 0.806061 & 0.8448 & 0.4048 & 0.2000 & 0.2000 \\
        \textbf{Mean} & 0.87941 & 0.622523 & 0.7635 & 0.6974 & 0.3281 & 0.2175 \\
        \textbf{3rd Qu.} & 0.87909 & 0.980559 & 0.9560 & 0.6934 & 0.6000 & 0.4000 \\
        \textbf{Max.} & 21.99891 & 0.999999 & 1.0000 & 17.5477 & 1.0000 & 1.0000 \\
        \textbf{NA's} & & 16 & & & & \\
        \hline
    \end{tabular}
\end{table}

\end{document}